\documentclass[a4paper,11pt]{article}

\usepackage[T1]{fontenc}
\usepackage[utf8]{inputenc}

\usepackage[margin=1in]{geometry}

\usepackage{microtype}

\usepackage{authblk}

\usepackage{mathpazo}

\usepackage{csquotes}

\usepackage{setspace}

\usepackage[rm]{titlesec}

\titlelabel{\thetitle.\enspace} 

\titleformat*{\section}{\Large\itshape} 
\titleformat*{\subsection}{\large\itshape} 
\titleformat*{\subsubsection}{\normalsize\itshape} 
\titleformat*{\paragraph}{\itshape} 

\usepackage[inline]{enumitem}
\setlist{noitemsep}

\usepackage[bottom, multiple]{footmisc}

\usepackage[labelfont=sc, labelsep=period, font=small, textfont=it]{caption} 

\usepackage[title, titletoc]{appendix} 

\usepackage{threeparttable}
\usepackage{booktabs}
\usepackage{tabularx}
\usepackage{colortbl}
\usepackage{makecell}

\usepackage{floatrow}
\usepackage{array}
\newcommand{\PreserveBackslash}[1]{\let\temp=\\#1\let\\=\temp}
\newcolumntype{C}[1]{>{\PreserveBackslash\centering}m{#1}}

\newcolumntype{s}{>{\hsize=.5\hsize}X}

\DeclareMathSizes{12}{12}{9}{9}

\usepackage{amsmath, amsthm, amsfonts}
\usepackage{mathtools}
\usepackage{bm}
\usepackage{nicefrac}
\usepackage{esvect}  

\newcommand{\wt}[1]{\widetilde{#1}}
\newcommand{\wh}[1]{\widehat{#1}}
\newcommand{\ul}[1]{\underline{#1}}
\newcommand{\ol}[1]{\overline{#1}}

\theoremstyle{plain}
\newtheorem{thm}{Theorem}
\newtheorem{prop}[thm]{Proposition}

\newtheorem{cor}[thm]{Corollary}
\newtheorem{lem}[thm]{Lemma}

\theoremstyle{definition}

\newtheorem{approximation}[thm]{Approximation}

\makeatletter
\newcommand*{\defeq}{\mathrel{\rlap{%
                     \raisebox{0.3ex}{$\m@th\cdot$}}%
                     \raisebox{-0.3ex}{$\m@th\cdot$}}%
                     =}
\makeatother

\makeatletter
\newcommand*{\invdefeq}{=\mathrel{\rlap{%
											  \raisebox{0.3ex}{$\m@th\cdot$}}%
                        \raisebox{-0.3ex}{$\m@th\cdot$}}%
                        }
\makeatother

\DeclareMathOperator*{\argmin}{arg\,min}

\newcommand{\E}{\mathbb{E}}
\newcommand{\R}{\mathbb{R}}
\newcommand{\me}{\mathrm{e}}
\newcommand{\diff}{\,\mathrm{d}}
\newcommand{\co}{\textsc{co}\oldstylenums{2}}

\usepackage{mathrsfs}

\usepackage[table,x11names, dvipsnames]{xcolor}
\usepackage{graphicx}

\usepackage{pgfplots}
\usepackage{pgfplotstable}

\pgfplotsset{width=0.65\textwidth, height=0.65*0.618\textwidth, compat=1.15}
\usepgfplotslibrary{fillbetween}

\usepackage{caption}
\usepackage{subcaption}

\usepackage[separate-uncertainty=true]{siunitx}

\DeclareSIUnit{\annum}{yr}
\DeclareSIUnit{\year}{yr}
\DeclareSIUnit{\years}{yr}
\DeclareSIUnit{\cent}{cent}
\DeclareSIUnit{\cents}{cents}
\DeclareSIUnit{\euro}{\texteuro}
\DeclareSIUnit{\euros}{\texteuro}
\DeclareSIUnit{\co}{\textsc{co}\oldstylenums{2}}
\DeclareSIUnit{\tco}{t_{\textsc{co}\oldstylenums{2}}}
\DeclareSIUnit{\euroco}{\texteuro/t\textsc{co}\oldstylenums{2}}
\DeclareSIUnit{\dollarco}{\textdollar/t\textsc{co}\oldstylenums{2}e}
\DeclareSIUnit\tC{tC}
\DeclareSIUnit\kWh{kWh}

\usepackage{eurosym}
\usepackage{icomma} 

\usepackage{hyperref}

\hypersetup{
    colorlinks,
    linkcolor={red!50!black},
    citecolor={black},
    urlcolor={blue!60!black}}

\usepackage[round]{natbib}

\title{
    Optimal Consumption and Investment with Energy-Efficiency Adoption
}

\author[a]{Anthony Britto\thanks{Corresponding author. Email: \href{mailto:anthony.britto@kit.edu}{anthony.britto@kit.edu}.}}

\author[b]{Carlos Oliveira\thanks{Email: \href{mailto:carlosoliveira@iseg.ulisboa.pt}{carlosoliveira@iseg.ulisboa.pt}.}}

\author[a]{Max Kleinebrahm\thanks{Email: \href{mailto:max.kleinebrahm@kit.edu}{max.kleinebrahm@kit.edu}.}}

\affil[a]{\normalsize Institute for Industrial Production, Karlsruhe Institute of Technology\protect\\ Hertzstr.~16 -- Building 06.33, 76187 Karlsruhe, Germany}

\affil[b]{\normalsize ISEG Lisbon School of Economics and Management,\protect\\ Departamento de Matemática -- Gab 303, Q2,\protect\\ Rua do Quelhas, nº 6, 1200-781 Lisboa, Portugal}

\begin{document}

\maketitle

\onehalfspacing

 \begin{abstract}
 	\noindent Despite many decades of research, economically grounded models that analyse energy consumption and energy-efficiency adoption within a unified framework remain underdeveloped. This article addresses this gap by proposing a model of consumption, investment, and energy-efficiency adoption under uncertainty. It develops new definitions of the rebound and backfire effects, and integrates their welfare implications into a model of optimal subsidy design. Macro-level technology diffusion and energy consumption across heterogeneous agents are also formalised. Explicit results for core objects are derived, including the adoption threshold and post-adoption strategies, and these are shown to depend on agent wealth, introducing a novel channel through which financial conditions influence technology-adoption decisions. An approximation scheme is proposed to estimate welfare implications explicitly. Adoption of energy efficiency is shown to be welfare improving in the main. A detailed case study of a representative German single-family home illustrates the theoretical results. Numerical analysis indicates that the subsidy policy effectively steers aggregate energy consumption.

    \bigskip
	
	
	\bigskip
	
\end{abstract}

\newpage

\section{Introduction}
\label{1d:sec:intro}

Beginning with the work of \citet{Hausman1979}, a substantial literature has documented the persistent puzzle of the slow adoption of cost-effective energy-efficiency measures. Despite decades of empirical and theoretical contributions, the evidence continues to indicate that households and firms often forgo ostensibly profitable investments, a phenomenon that has come to be known as the \emph{energy-efficiency gap} \citep{Jaffe1994a}. In a recent review of the economic literature, \citet{Gerarden2017} identify three broad categories of explanations for the apparent underinvestment in energy-efficiency technologies: market failures, behavioural explanations, and modelling flaws. Broadly speaking, market-failure explanations emphasise frictions that prevent private actors from fully capturing the benefits of energy-efficiency investments; these range from information asymmetries and principal–agent problems to capital-market constraints and mispriced energy. Behavioural explanations, by contrast, highlight systematic patterns in decision-making that lead individuals to deviate from economically rational choices; examples of such explanations include limited attention, reliance on heuristics, short planning horizons, and biased beliefs. We focus on the third category of explanations, namely modelling flaws, which offers reasons as to why observed rates of energy-efficiency investment are not as paradoxical as they first appear. The underlying motivation is well summarised by \citet{Gillingham2014}:
\begin{quote}
    Many economists believe that consumer choices reveal more about the economics of energy efficiency improvements than do engineering calculations. If engineering estimates of the energy savings potential from seemingly cost-effective investments fail to include some costs or model the consumer’s decision inappropriately, then the assessment of what is optimal from the consumer’s perspective will be incorrect. Thus the engineering approach will result in the net benefits from energy efficiency investments being overstated, which means the [energy-efficiency] gap may be much smaller than estimated or there may be no gap at all\ldots More broadly, the gap may be overestimated because of hidden costs, consumer heterogeneity, uncertainty, overestimated savings, and the rebound effect.
\end{quote}
We propose using models of optimal consumption and investment under uncertainty as a framework for this discussion. This article develops such a decision model and shows that it naturally incorporates many of the factors cited above that are of particular economic interest, including uncertainty, consumer heterogeneity, and the rebound effect. Moreover, we demonstrate how the model provides a natural basis for analysing the welfare implications of consumer decisions. Specifically, this article makes four main contributions, which together illustrate the analytical and policy relevance of the framework.

First, the decision model yields closed-form solutions for key quantities, including the energy-efficiency adoption threshold and post-adoption optimal strategies. These are shown to depend explicitly on wealth, introducing a novel channel through which financial conditions influence technology adoption decisions. Second, we develop new definitions of the rebound and backfire effects and associated welfare implications, and demonstrate how these may be estimated under suitable approximations. We show that energy-efficiency adoption is indeed welfare improving in the main, thereby providing analytical support for a central policy intuition. Third, the analysis is extended to optimal subsidy design, examining how externalities from energy consumption can be mitigated in the absence of Pigouvian taxation. Fourth, by embedding heterogeneous agents in the proposed framework, it is shown how macro-level technology diffusion patterns emerge endogenously from micro-level optimisation in a stochastic environment. The interaction between uncertainty and heterogeneity generates non-linear diffusion effects, which can be partly offset by an appropriately designed subsidy policy that stabilises adoption incentives. These theoretical results are illustrated through an in-depth case study of a representative German single-family home, which demonstrates the model's applicability and quantitative relevance.

We mention some related work. A strand of macroeconomic dynamics, exemplified by \citet{Pommeret2009, Khan2002}, studies the socially-optimal adoption of cleaner or more efficient technologies in a general-equilibrium setting. In these models, households act as producer–consumers, transforming capital into consumption, and optimal Pigouvian taxation is used to decentralise the social optimum. The focus is therefore on economy-wide efficiency and policy design. By contrast, financial markets are typically abstracted away, and the role of individual risk and borrowing constraints is limited. This literature thus does not capture the interaction between financial conditions, uncertainty, and energy-service demand that is central to the present analysis, which instead takes its cue from the literature on the microeconomics of the rebound effect \citep[cf.][]{Borenstein2015, Chan2015, Hunt2015}. In particular, the model developed below accounts for borrowing, saving, and price dynamics, yielding a richer characterisation of the intertemporal allocation of financial resources for utility maximisation.


In addition to the macro-dynamics literature, our model relates to research on real options in energy-efficiency adoption. Beyond the seminal work by \citet{Hassett1993}, we mention the studies by \citet{Britto2024, Tadeu2016, Lee2014, Kumbaroglu2012}. It is worth noting that this literature generally adopts an \enquote{investment perspective}, focusing on the minimisation of energy costs, whereas our approach follows an theoretically grounded utility-maximisation perspective. The model also has connections to the literature on consumption-investment models with subsistence constraints \citep{Jeon2020, Achury2012, Sethi1992}. In particular, \citeauthor{Achury2012} show that imposing a time-invariant subsistence requirement can reproduce key empirical regularities in household behaviour, including wealth-dependent saving rates, portfolio allocations, and effective risk aversion. These results support our framework and motivate its extension to a two-good setting.


This article focuses on space heating in residential buildings due to the scale and significance of the sector. This scope also facilitates comparison with the literature, where consumer decisions regarding energy use and improvements in building energy performance, commonly referred to as \enquote{retrofits}, are studied extensively. One reason for this is the size of the sector: globally, roughly \qty{5}{\percent} of greenhouse gas emissions can be traced back to direct on-site emissions from residential buildings due to space heating \citep{Cabeza2022}. The other reason is that the building sector, despite having enormous potential for energy savings according to engineering estimates, has shown relatively slow progress in energy-efficiency adoption despite extensive policy interventions \citep{Nejat2015}. Germany, the focus of our case studies, aims to achieve overall climate neutrality by 2045, and has set the goal of reducing emissions in the building sector to \qty{57}{\percent} of 2020 levels by 2030 \citep{Bundesregierung2021}. The gap between these stated policy goals and actual retrofit rates is significant and well-documented. For instance, the German Energy Agency, a consultancy, estimates that the rate of deep retrofitting needs to roughly double, from the historic level of \qty{1}{\percent\per\year} to around \qty{1.9}{\percent\per\year}, in order to achieve these goals \citep{Jugel2021}.

The remainder of the article is structured as follows. Section~\ref{1d:sec:model} develops the model, including the agent's decision problem, welfare implications, optimal subsidy design, and aggregate energy-efficiency uptake. Section~\ref{1d:sec:solution} presents the solution, including explicit results for the agent's optimal strategies and approximate results for many of the welfare quantities. The model is applied to a case study in Section~\ref{1d:sec:case_study} following, which demonstrates the plausibility of the model and presents key comparative statics. Section~\ref{1d:sec:discuss} concludes.

\section{The model}
\label{1d:sec:model}

The model is built around the agent’s decision problem of optimal consumption, investment, and energy-efficiency adoption, which we present first. Section~\ref{1d:sec:welfare:definitions} then analyses the welfare implications of these decisions, including the problem of optimal subsidy design to mitigate externalities from energy consumption. Lastly, Section~\ref{1d:sec:diffusion} turns to the macro perspective by defining key aggregate quantities over a population of heterogeneous agents.

\subsection{The agent's decision problem}
\label{1d:sec:decision_problem}

There are two consumption goods: the energy service \enquote{heating} $s$, and a perishable non-energy good $x$, the numeraire. The energy service is obtained from fuel consumption $c$ according to $s = \eta c$, where $\eta > 0$ quantifies the total efficiency of the energy-conversion chain. The agent may, at any $\tau \in \mathcal{T}$ where $\mathcal{T}$ is the set of positive stopping times, choose to retrofit their dwelling at cost $K > 0$ to an efficiency-level $\wt{\eta} > \eta$. We specialise to the setup where $K$ is large and assume the agent finances the investment through a loan. For simplicity, we assume that the loan is serviced indefinitely by a constant payment flow $\rho K$, where $\rho > 0$ is the borrowing rate. To avoid confusion, the term \enquote{investment} is hereafter used exclusively to refer to the investment in the retrofit, with the term \enquote{allocation} being reserved for the agent's portfolio decisions.

The price of energy $P > 0$ is assumed constant, as is the rate of labour income $Y > 0$. There are two financial assets: a risk-free bond with price $S^0_t$ and dynamics
\begin{equation}
    \diff S^0_t = \mu_R S^0_t \diff t\ ,
\end{equation}
where $\mu_R > 0$, and an index fund with price $S_t$ whose dynamics follow
\begin{equation}
    \diff S_t = \mu_S S_t \diff t + \sigma_S S_t \diff B_t\ ,
\label{eq:risky_asset}
\end{equation}
where $\mu_S >\mu_R$, $\sigma_S > 0$, and $B_t$ is a standard Brownian motion. At any time $t$, the agent invests a share $a_t$ of wealth in the index fund, with the remainder allocated to the bond. For realism, we assume that the borrowing rate $\rho$ from above satisfies $\rho \geq \mu_R$. The market price of risk is denoted $\kappa \defeq (\mu_S - \mu_R) / \sigma_S$.

For a given $\tau \in \mathcal{T}$, it follows that the agent's wealth is given by
\begin{multline}
    \diff W_t = \frac{a_t W_t}{S_t} \diff S_t + \frac{(1 - a_t) W_t}{S^0_t} \diff S^0_t + (Y - x_t) \diff t\\
    - \left( (s_t/\eta) P\,\mathbb{1}_{\{t < \tau\}} + ((s_t/\wt{\eta}) P + \rho K) \mathbb{1}_{\{t \geq \tau\}} \right) \diff t\ , \quad t \geq 0\ ,
\end{multline}
which simplifies to
\begin{multline}
    \diff W_t = ( a_t \mu_S W_t + (1 - a_t) \mu_R W_t ) \diff t + a_t \sigma_S W_t \diff B_t + (Y - x_t) \diff t\\
    - \left( (s_t/\eta) P\,\mathbb{1}_{\{t < \tau\}} + ((s_t/\wt{\eta}) P + \rho K) \mathbb{1}_{\{t \geq \tau\}} \right) \diff t\ .
\label{1d:eq:wt_evolve}
\end{multline}
Let $\{ W_t^{w; \tau} \mid t \geq 0 \}$ denote a solution to \eqref{1d:eq:wt_evolve} for a given initial condition $W_0 = w$ and investment time $\tau$. Assuming time-additive utility with exponential discounting, the agent's decision problem is to choose an allocation control, consumption controls, and an investment time $\tau$ to maximise the present value of utility up to a random time horizon $T > 0$. The random time horizon is intended to capture exogenous shocks to the consumption-investment process such as the purchase or sale of a house, a sudden change in employment status, or death. Thus, for a given initial wealth $w$, the agent's value function is written as
\begin{equation}
    F(w) \defeq \sup_{a, x, s, \tau} \E \left[
    \int_0^T \me^{-\delta t} U (x_t, s_t) \diff t \ \Big|\ W_0 = w
    \right]\ ,
\label{1d:eq:decision_prob_init}
\end{equation}
where $\delta > 0$ is the rate of time preference and $U \colon \R^2_+ \to \R$ is the utility function.

Following \citet{Merton1971}, assume that $T$ follows an exponential distribution with cumulative distribution function $\mathcal{F}_T(t) = 1 - \me^{-\lambda t}$, where $\lambda > 0$ is the hazard rate; assume further that $T$ is independent of the securities market and the controls. The expectation over the random time can then be rewritten as a weighted integral over all $t$ as follows:
\begin{align}
    \E \left[ \int_0^T \me^{-\delta t} U(x_t, s_t) \diff t \right] &= \E \left[ \int_0^\infty \me^{-\delta t}U(x_t, s_t) \mathbb{1}_{\{t < T \}} \diff t \right]\\
    &= \E \left[ \int_0^\infty \me^{-\delta t}U(x_t, s_t) (1 - \mathcal{F}_T(t)) \diff t \right]\\
    &= \E \left[ \int_0^\infty \me^{-(\delta + \lambda) t} U(x_t, s_t) \diff t \right]\ . \label{1d:eq:random_time_discount}
\end{align}
We hence reduce to a problem of optimal control and stopping on an infinite horizon:
\begin{equation}
    F(w) \defeq \sup_{a, x, s, \tau} \E \left[
    \int_0^\infty \me^{-\wh{\delta} t} U (x_t, s_t) \diff t \ \Big|\ W_0 = w
    \right]\ ,
\label{1d:eq:decision_prob}
\end{equation}
where $\wh{\delta} \defeq \delta + \lambda$ is the agent's effective discount rate. The remainder of this article assumes Stone-Geary preferences,
\begin{equation}
    U(x, s) \defeq \frac{\left( (x - \ul{x})^{1 - \beta} (s - \ul{s})^\beta \right)^{1-\gamma}}{1-\gamma}\ ,
\end{equation}
where $\ul{x} > 0$ and $\ul{s} > 0$ are the subsistence levels of the two goods, $\beta \in (0, 1)$ the relative preference weight on the energy service, and $\gamma > 1$ the coefficient of risk aversion. This specification retains the mathematical tractability of the Cobb-Douglas function while adding the realistic feature that a minimum level of consumption is required for survival.


\subsection{Welfare implications}
\label{1d:sec:welfare:definitions}

Having introduced the agent's decision problem, we turn to the welfare implications of their consumption and technology-adoption choices. In sequence, we develop definitions for the rebound and backfire effects and associated welfare implications, and conclude by formulating a problem of optimal subsidy design in a second-best world.

Classically, the rebound effect is studied in the context of a standard two-good utility-maximisation problem, where it is defined as the elasticity of energy-service demand with respect to efficiency \citep[cf.][]{Borenstein2015, Chan2015}. Such an approach is not possible in our model since the optimal consumption path is given by a stochastic process with a possible discontinuity when investment occurs. We therefore propose an alternative definition that aligns with the intuition underlying the rebound effect, which can be expressed as the following question: \enquote{How does energy-service demand after an energy-efficiency investment compare to a counterfactual scenario, where no investment occurs?} This type of counterfactual question is standard in welfare economics, where behavioural responses are assessed by comparing actual outcomes with a clearly defined status-quo benchmark \citep[cf.][Ch.~10]{Varian1992}.

To this end, consider the no-investment limit $\tau = \infty$ in \eqref{1d:eq:wt_evolve}; we denote this process by $\wh{W}_t$. It follows the dynamics
\begin{equation}
    \diff \wh{W}_t = (\wh{a}_t \mu_S \wh{W}_t + (1 - \wh{a}_t)\mu_R \wh{W}_t + Y - \wh{x}_t - (\wh{s}_t/\eta) P) \diff t + \wh{a}_t \sigma_S \wh{W}_t \diff B_t\ ,
\label{1d:eq:counterfactual_value}
\end{equation}
where the allocation and consumption controls also carry hats for clarity. Let $\{ \wh{W}_t^w \mid t \geq 0 \}$ denote a solution to \eqref{1d:eq:counterfactual_value} for a given initial condition $\wh{W}_0 = w$. Thus, for a given initial wealth $w$, the counterfactual decision problem is given by
\begin{equation}
    \wh{F}(w) \defeq \sup_{\wh{a}, \wh{x}, \wh{s}} \E \left[ \int_0^\infty \me^{-\wh{\delta} t} U (\wh{x}_t, \wh{s}_t) \diff t \ \Big|\ \wh{W}_0 = w \right]\ .
\label{1d:eq:counterfactual_problem}
\end{equation}
Consequently, the difference in energy-service demand between the original decision problem \eqref{1d:eq:val_func} and the counterfactual is measured by the process
\begin{equation}
    R_t \defeq s_t^* - \wh{s}_t^*\ ,
\label{1d:eq:def_rebound}
\end{equation}
where the asterisks indicate that the controls are optimal in their respective problems. It is then natural to say that \emph{rebound occurs} if $R_t > 0$ for $t \geq \tau^*$, i.e.~if energy-service demand after investment is greater than in the counterfactual. Note that any rebound observed in this context is \enquote{optimal} from the agent's perspective, with the definition naturally incorporating the utility-maximising levels of consumption, the optimal investment time $\tau^*$, and the effects of the retrofit cost on the budget constraint.

The related notion of \emph{backfire}, whereby net fuel savings vanish due to excessive rebound, can be similarly formalised. Defining the process
\begin{equation}
    Q_t \defeq
    \begin{cases}
         s_t^*/\eta - \wh{s}_t^*/\eta\ ,&t < \tau^*\ ,\\
         s_t^*/\wt{\eta} - \wh{s}_t^*/\eta\ ,&t \geq \tau^*\ ,
    \end{cases}
\label{1d:eq:def_backfire}
\end{equation}
we say that \emph{backfire occurs} if $Q_t > 0$ for $t \geq \tau^*$. A key point, which is exploited in the sequel, is that it is possible to have rebound without backfire. Indeed, for $t \geq \tau^*$, $Q_t \leq 0$  is equivalent to imposing $s^*_t \leq \wt{\eta} \wh{s}^*_t / \eta$, which allows for some level of rebound since $\wt{\eta} > \eta$. For instance, if $\eta = 0.5$ and $\wt{\eta} = 0.7$, we have $s^*_t \leq 1.4\,\wh{s}^*_t$; hence, backfire occurs only if energy-service demand after the investment increases by more than \qty{40}{\percent} relative to the counterfactual.

We make two additional remarks. Firstly, since $R_t$ and $Q_t$ are stochastic processes, our definitions of rebound and backfire are probabilistic. It is possible to remove the randomness by computing expected rebound (resp.~expected backfire) at time $t$, $\E [R_t]$ (resp.~$\E [Q_t]$). Secondly, the definitions of $R_t$ and $Q_t$ are time-dependent, so that changes in wealth affect the level of rebound and backfire. For instance, if the agent becomes wealthier in expectation as time passes, the probability of rebound or backfire occurring increases in line with the increase in overall spending (cf.~Proposition~\ref{1d:prop:rebound} below). To the best of our knowledge, this is the first formalisation in the literature of the time dependence of the rebound and backfire effects. The importance of these time effects becomes apparent when considering the welfare implications of the energy-efficiency investment, a task which we now take up.

The total welfare change from the retrofit has two components: the agent’s private gain and the change in the social cost, each measured relative to the counterfactual. We focus here on the social cost; the agent’s welfare is discussed in Appendix~\ref{append:private_gain}. As shown there, the agent’s welfare is always at least as high as in the counterfactual; accordingly, the more informative margin is the social component, where non-trivial effects arise. All other things being equal, the social cost of the retrofit is driven by externalities from energy consumption.\footnote{\label{1d:foot:energy_service}In addition to direct externalities from energy consumption such as greenhouse gas emissions, \cite{Chan2015} observe that there may also be externalities tied to the energy-\emph{service} itself, e.g.~congestion externalities from excessive driving. It is straightforward to extend the proposed framework to account for this, though we do not do so since the energy service \enquote{heating} generates no material negative spillovers beyond energy use, and any incidental effects, e.g.~neighbouring units benefiting from excess heat, are either positive or negligible relative to the externalities from fuel consumption.} Hence, the present value of social costs in the presence of a retrofit may be written as
\begin{equation}
    \E \left[
    \int_0^{\tau^*} \me^{-\wh{\epsilon} t} \left( \varpi_t^\pi\,s^*_t / \eta \right) \diff t
    +
    \int_{\tau^*}^\infty \me^{-\wh{\epsilon} t} \left( \varpi_t^\pi\, s^*_t / \wt{\eta}
    \right) \diff t\
    \right]\ ,
\label{1d:er:sw_1}
\end{equation}
where $\wh{\epsilon} \defeq \epsilon + \lambda > 0$ is the sum of the social discount rate $\epsilon > 0$ and hazard rate $\lambda$ as in \eqref{1d:eq:val_func}, and $\varpi_t > 0$ is the marginal social cost of energy consumption with initial condition $\varpi_0 = \pi > 0$. We assume that the marginal social cost $\varpi_t$ follows a geometric Brownian motion, independent of the wealth process, with drift $\mu_\varpi > 0$ and volatility $\sigma_\varpi > 0$. In the counterfactual, the social cost is given by
\begin{equation}
    \E \left[ \int_0^\infty \me^{-\wh{\epsilon} t}
    \left(
    \varpi_t^\pi\, \wh{s}^*_t/\eta
    \right) \diff t \right]\ .
\label{1d:er:sw_2}
\end{equation}
Subtracting \eqref{1d:er:sw_2} from \eqref{1d:er:sw_1}, we see that the change in social cost due to the energy-efficiency investment can be expressed succinctly in terms of the backfire measure as
\begin{equation}
    V_\text{sc}(w; \pi) \defeq \E \left[ \int_0^\infty \me^{-\wh{\epsilon} t} \varpi_t^\pi Q_t^w \diff t \right]\ ,
\label{1d:eq:social_cost}
\end{equation}
where we make explicit that the backfire measure depends on initial wealth $w$. It follows immediately that if backfire does not occur, i.e.~if $Q_t \leq 0$ for all $t$, then the social cost is non-positive, implying a social gain. A direct implication of this result is that mitigation policies should focus on backfire rather than rebound. Indeed, as long as rebound remains below the level which leads to backfire, the agent's utility gains, which derive in part from rebound, can be preserved without driving social welfare below zero. For this reason, policies that aim to suppress rebound per se are at least partly misguided, as has been noted in the literature \citep{Borenstein2015, Gillingham2014}.\footnote{With respect to Footnote~\ref{1d:foot:energy_service}, these conclusions are valid only if energy-service demand does not generate externalities. If this were not the case, the expression for $V_\text{sc}$ in \eqref{1d:eq:social_cost} would in fact contain a term including the rebound measure $R_t$, making rebound a legitimate policy target.}




\subsection{Design of a subsidy policy}

A natural extension of the above considerations is the study of corrective policies for energy-consumption externalities. Since we are in a partial-equilibrium setup, a full characterisation of the design of optimal Pigouvian taxes lies outside our scope.\footnote{Briefly, a Pigouvian tax alters the effective energy price faced by \emph{all} agents in the economy. As a result, one must account not only for heterogeneous adoption responses, but also for the distributional consequences of higher energy prices, the reallocation of expenditure across consumption goods, and the potential general-equilibrium feedbacks on wages, capital returns, and output. In other words, the welfare accounting of a tax cannot be reduced to a localised transfer problem, but requires an explicit aggregation of utility losses and externality reductions across the entire population. This type of analysis is proper to a general-equilibrium framework, where heterogeneity across agents and market interactions can be made explicit.} Given this limitation, we assume a second-best world in which corrective taxation is unavailable, and consider a social planner who addresses consumption externalities by subsidising the retrofit at rate $m \in \R$, reducing the agent's private cost to $(1 - m)K$.\footnote{We underline that the typical role of subsidies is to correct for investment inefficiencies (e.g.~behavioural distortions), which do not exist in our model since the agent is a rational expected-utility maximiser \citep[cf.][]{Allcott2016, Allcott2012}. Note further that since Pigouvian taxes are absent, it is possible that the subsidy rate may be negative, i.e.~a penalty may be imposed. This is justified in the event of consumption backfire; cf.~Proposition~\ref{1d:prop:approx_subsidy}.} Consequently, the agent's investment time $\tau^*$, as well as the allocation and consumption rules are altered. Assuming that the subsidy is paid when the agent invests, the social planner's objective of minimising social costs is written as
\begin{equation}
    J(w; \pi) \defeq \inf_m \E \bigg[ V_{\text{sc}, m}(w; \pi) + \me^{-\epsilon \tau^*_m(w)} \Psi(m K) \bigg]\ ,
\label{1d:eq:opt_subsidy_def}
\end{equation}
where
\begin{equation}
    \Psi(x) \defeq \xi_0 x + \frac{\xi_1}{2} x^2
\label{1d:eq:planner_cost_func}
\end{equation}
is a convex cost function.\footnote{We write $V_{\text{sc}, m}$ to indicate the influence of the subsidy on social cost through the channels of the agent's energy consumption and investment timing.} The parameter $\xi_0 \geq 1$ can be identified with the marginal cost of public funds \citep{Browning1976}, with $\xi_1 > 0$ an additional friction parameter to disincentivise large transfers. This is a standard bilevel or Stackelberg optimisation, since the planner anticipates the agent's response and chooses the subsidy accordingly \citep{Colson2007}. Notice that the optimal subsidy derived here is tailored to the characteristics of the individual agent; in particular, its level depends on the initial endowment $w$ and income $Y$, dimensions which are essential to understanding free-riding behaviour \citep{Rivers2016, Nauleau2014}.

\subsection{Modelling aggregate behaviour}
\label{1d:sec:diffusion}

The model concludes by examining how aggregate energy-efficiency adoption and energy consumption evolve over time. The resulting curves may be used to benchmark the energy-efficiency gap, and for policy analysis \citep[cf.][]{Hassett1993}. In the present models, the primary force shaping both curves is the interaction between agent heterogeneity and the stochastic environment.

Given our setup, it is natural to assume that the agents face identical conditions on the financial market, so that the risk-free rate $\mu_R$ and the risky-asset parameters $\mu_S$ and $\sigma_S$ are common to all agents. It is also natural to assume that the agents face the same price of energy $P$. The remaining parameters are then idiosyncratic to each agent. Firstly, there are the preference parameters $\beta$, $\gamma$, $\delta$, and $\lambda$, and the subsistence consumption levels $\ul{x}$ and $\ul{s}$. Then there are the retrofit parameters $\rho$, $K$, $\eta$, and $\wt{\eta}$. Finally we have labour income $Y$ and initial wealth $w$. By drawing $N$ times from an assumed joint distribution for these parameters, we generate a population of representative agents
\begin{equation}
    \mathcal{P} = \{ (\beta^i, \gamma^i, \delta^i, \ldots, Y^i, w^i) \mid i = 1, \dots, N \}\ .
\label{1d:eq:population}
\end{equation}
Conditional on this population, the \emph{share of adopters} at time $t$ follows the stochastic process
\begin{equation}
    S_t \defeq \frac{1}{N} \sum_{i=1}^N \mathbb{1}_{\{\tau^i(w^i) \leq t\}}\ ,
\label{dc:eq:share}
\end{equation}
where $\tau^i = \tau^i(w^i)$ is the optimal investment time of each agent. Intuitively, for a given realisation of the financial market, the process $S_t$ experiences jumps of size $1/N$ each time an agent invests. Its expectation, $\E[S_t]$, is the main quantity of interest, denoting the expected uptake, or \emph{diffusion}, of the energy efficiency measure. Similarly, aggregate energy consumption is given by
\begin{equation}
    C_t \defeq \sum_{i = 1}^N \left[ (s_t / \eta) \mathbb{1}_{\{\tau^i > t\}} + (s_t / \wt{\eta}) \mathbb{1}_{\{\tau^i \leq t\}}  \right]\ , 
\label{dc:eq:tot_energy}
\end{equation}
where the efficiency parameter is upgraded from $\eta$ to $\wt{\eta}$ following the investment.
Since the planner aims to minimise externalities from energy consumption, the quantity $\mathbb{E}[C_t]$ is in fact an indirect target of the subsidy policy considered in Section~\ref{1d:sec:welfare:definitions} above. By explicitly accounting for the heterogeneity underlying these dynamics, the subsidy allows the social planner to steer aggregate energy consumption more effectively by targeting energy-efficiency adoption where it generates the greatest marginal benefit.

\section{The solution}
\label{1d:sec:solution}

Due to the assumptions of constant prices and wages, the agent's decision problem admits an almost complete closed-form solution, which we present first. Then in Section~\ref{1d:sec:welfare:approx}, an approximation is introduced to derive explicit solutions for the welfare aspects of the model. Analogously, Section~\ref{1d:sec:subsidy:approx} presents the approximate optimal subsidy policy. Finally, Section~\ref{1d:sec:sol:aggregate} considers the aggregate quantities of interest.

\subsection{The agent's optimal strategies}
\label{1d:sec:sol:opt_strategies}

We begin by simplifying the decision problem in \eqref{1d:eq:decision_prob}. To this end, it will be helpful to have a separate notation for the agent's wealth in the case of immediate investment, i.e.~$\tau = 0$. Denoting this process by $\wt{W}_t$, we write down its dynamics from \eqref{1d:eq:wt_evolve} as
\begin{equation}
    \diff \wt{W}_t = \left (\wt{a}_t \mu_S \wt{W}_t + (1 - \wt{a}_t)\mu_R \wt{W}_t + \wt{Y} - \wt{x}_t - (\wt{s}_t/\wt{\eta}) P \right) \diff t + \wt{a}_t \sigma_S \wt{W}_t \diff B_t\ ,
\label{1d:eq:w_tilde_t}
\end{equation}
where $\wt{Y} \defeq Y - \rho K$ is labour income net of the loan payment. For clarity, the allocation and consumption controls here have also been labelled with a tilde. Then, noticing the structure of the controlled dynamic in \eqref{1d:eq:wt_evolve}, we make use of the strong Markov property of geometric Brownian motion and the law of total expectation to rewrite \eqref{1d:eq:decision_prob} as a problem of optimal control and stopping on an infinite horizon:
\begin{equation}
    F(w) = \sup_{a, x, s, \tau} \E \bigg[ \int_0^\tau \me^{-\wh{\delta} t} U(x_t, s_t) \diff t + \me^{-\wh{\delta} \tau} G(W_\tau^{w;\tau})\ \Big|\ W_0 = w \bigg]\ ,
\label{1d:eq:val_func}
\end{equation}
with
\begin{equation}
    G(w) \defeq \sup_{\wt{a}, \wt{x}, \wt{s}} \E \bigg[ \int_0^\infty \me^{-\wh{\delta} t} U(\wt{x}_t, \wt{s}_t) \diff t\ \Big|\ \wt{W}_0 = w \bigg]
\label{1d:eq:def_g}
\end{equation}
being the value function conditional on immediate investment. With a slight abuse of terminology, we refer to $G$ as the \emph{terminal gain} in the sequel.

The function $G$ is a model of optimal consumption and allocation with two goods and subsistence constraints over an infinite horizon. As such, it is an interesting and relevant extension to the literature on subsistence constraints discussed in Section~\ref{1d:sec:intro}. The typical first step in solving optimal control problems with labour income is to calculate human capital \citep[cf.][]{Bensoussan2025, Kraft2011}. In this instance, it is given by the present value of effective labour income net of subsistence consumption, i.e.
\begin{equation}
    \wt{H} \defeq \int_0^\infty \me^{-\mu_R t} (\wt{Y} - \ul{x} - (\ul{s} / \wt{\eta}) P) \diff t = \frac{1}{\mu_R} (\wt{Y} - \ul{x} - (\ul{s} / \wt{\eta}) P )\ .
\label{1d:eq:def_human_capital}
\end{equation}
Since the agent is allowed to borrow against human capital, define the total money available for discretionary spending as
\begin{equation}
    \wt{Z}_t \defeq \wt{z}(\wt{W}_t) \defeq \wt{W}_t + \wt{H}\ .
\end{equation}
For ease of terminology, we refer to $\wt{Z}_t$ as \enquote{disposable capital} in the sequel. For economic realism, this quantity must be constrained to be positive, else the agent could simply borrow against human capital indefinitely. Consequently, for a given initial condition $w > - \wt{H}$, define the set of admissible controls as
\begin{equation}
    \wt{\mathcal{A}}(w) = \left\{ (\wt{a}, \wt{x}, \wt{s}) \mid \wt{z}(\wt{W}_t^w) > 0 \ \forall\ t \geq 0 \right\}\ .
\label{eq:admit_control}
\end{equation}
The following result is obtained.

\begin{prop}
\label{1d:prop:terminal_gain}
    Let $w > -\wt{H}$. The terminal gain in \eqref{1d:eq:def_g} is given by
    \begin{equation}
        G(w) = \Gamma^{-\gamma} \frac{\wt{z}(w)^{1-\gamma}}{1-\gamma}\ ,
    \label{eq:g_expr}
    \end{equation}
    where $\Gamma > 0$ is a constant defined in \eqref{eq:def_gamma}. The optimal strategies are
    \begin{equation}
        \wt{a}^*_t = \frac{\kappa}{\gamma \sigma_S}\frac{\wt{Z}_t^{\wt{z}(w)}}{\wt{W}_t^{w}}\ ,\quad
        \wt{x}^*_t = \ul{x} + (1 - \beta) \varphi \wt{Z}_t^{\wt{z}(w)}\ ,\quad
        \wt{s}^*_t = \ul{s} + \beta \varphi \frac{\wt{Z}_t^{\wt{z}(w)}}{P/\wt{\eta}}\ ,
    \label{1d:eq:opt_controls_tg}
    \end{equation}
    where $\varphi > 0$ is a constant defined in \eqref{eq:def_varphi}.
\end{prop}

The proof, along with the other mathematical proofs for this article, is provided in Appendix~\ref{append:base_model:proofs}. A few remarks are in order. Notice firstly that in the limit $\wt{Z}_t \to 0$, allocation vanishes and consumption reduces to subsistence levels. On the other hand, allocation is always above the Merton level $\kappa/(\gamma \sigma_S)$, with this level being attained only asymptotically in the limit of large wealth, $\wt{W}_t \gg \wt{H}$. In the case of the consumption controls, optimal behaviour is straightforward to interpret: total expenditure on consumption above subsistence levels is given by $\varphi \wt{Z}_t$, with the sum being allotted to each good according to the preference weight $\beta$. Energy-service demand $\wt{s}^*_t$ is seen to depend on the so-called \enquote{implicit price} of energy, $P / \wt{\eta}$ \citep[cf.][]{Chan2015}. Due to the presence of the subsistence levels, price elasticity and income elasticity deviate from unity.



We note that in a model with exogenous labour income and no constraints or transaction costs on risky-asset allocation, it is possible for an impatient agent to borrow heavily against human capital, driving wealth into deeply negative regimes. To see this, note that a direct corollary of Proposition~\ref{1d:prop:terminal_gain} is that the optimally-controlled process $\wt{Z}_t$ follows a geometric Brownian motion
\begin{equation}
    \diff \wt{Z}_t = \left( \frac{\kappa ^2 + \gamma \left(\kappa ^2 -2 \wh{\delta} + 2 \mu_R \right)}{2 \gamma ^2}  \right) \wt{Z}_t \diff t + \frac{\kappa}{\gamma} \wt{Z}_t \diff B_t\ 
\label{1d:eq:opt_z_tilde_t}
\end{equation}
with solution
\begin{equation}
    \wt{Z}_t = \wt{z}(w) \exp{\left[ \left(\mu_{\wt{Z}} - \tfrac{1}{2} \sigma_{\wt{Z}}^2 \right)t + \sigma_{\wt{Z}} B_t \right]}\ ,
\end{equation}
where $\mu_{\wt{Z}}$ and $\sigma_{\wt{Z}}$ denote the drift and volatility respectively in \eqref{1d:eq:opt_z_tilde_t}. Hence, if $\mu_{\wt{Z}} - \sigma_{\wt{Z}}^2 /2 < 0$, the process $\wt{Z}_t$ shrinks in expectation over time, so that as $t$ becomes large, the agent exhausts human capital entirely by borrowing. Such unrealistic behaviour can be avoided by requiring that the effective drift remain positive; straightforward computation shows that this requirement reduces to the condition
\begin{equation}
    \wh{\delta} < \frac{\kappa^2 + 2 \mu_R}{2}\ .
\label{1d:eq:patience_cond}
\end{equation}
This \enquote{patience condition} is taken as given in the following.



We take up the agent's decision problem in full. As above, define the quantities
\begin{equation}
    Z_t \defeq z(W_t) \defeq W_t + H\ , \quad \text{where} \quad H \defeq \frac{1}{\mu_R} (Y - \ul{x} - (\ul{s} / \eta) P )\ .
\end{equation}
Then for $w > -H$, define the set of admissible controls analogously to \eqref{eq:admit_control} as
\begin{equation}
    \mathcal{A}(w) \defeq \{ (a, x, s, \tau) \mid z(W_t^{w;\tau}) > 0 \ \forall\ t \geq 0 \}\ . 
\end{equation}
Denote also the change in human capital due to adoption of the energy-efficiency measure as
\begin{equation}
    \theta \defeq \wt{H} - H = \frac{1}{\mu_R} \left( (\ul{s}/\eta - \ul{s}/\wt{\eta})P - \rho K \right)\ .
\end{equation}
Intuitively, the change simply equals the net present value of energy costs at subsistence levels. We hence refer to $\theta$ as the \enquote{subsistence net present value} in the sequel. It bears emphasising that in contrast to a net present value analysis based on average demand, e.g.~the cost-minimisation model in \citet{Hassett1992}, the above definition is based on the \emph{subsistence} demand, where the agent has no more flexibility. The solution to the agent's decision problem follows.

\begin{thm}
\label{1d:thm:val_func}
    Let $w > -H - \theta$ be a given initial wealth. Define the threshold $w^* \defeq \Lambda \theta - H$ for $\Lambda < 0$ a constant defined in \eqref{1d:eq:trig_zstar}, and let $z^* = z(w^*)$. The following cases are obtained.
    \begin{enumerate}
        \item Suppose $\theta \geq 0$ or $w \geq w^*$. The value function of \eqref{1d:eq:val_func} is given by $F(w) = G(w)$, with $\tau^*(w) = 0$ being optimal in \eqref{1d:eq:val_func}. The optimal allocation and consumption strategies are given in \eqref{1d:eq:opt_controls_tg}.
        \item Suppose $\theta < 0$ and $w < w^*$. The value function is given by
        \begin{equation}
            F(w) = \inf_{\wh{z} > 0} \left[ \wh{f}(\wh{z}) + \wh{z}\,z(w) \right]\ ,
        \end{equation}
        where $\wh{f}$ is defined in \eqref{1d:eq:f_hat}. The first hitting time
        \begin{equation}
            \tau^*(w) = \inf\{t \geq 0 \mid W_t^w \geq w^* \}
        \end{equation}
        is optimal in \eqref{1d:eq:val_func}.\footnote{We abuse notation slightly by using $W_t^w$ to denote a solution to \eqref{1d:eq:wt_evolve} with initial condition $W_0 = w$ for $t < \tau^*$.} The remaining optimal strategies follow
        \begin{equation}
            a^*_t =
            \begin{dcases}
                -\frac{\kappa}{\sigma_S}\frac{\partial_w F(W_t^w)}{w \partial_w^2 F(W_t^w)}\ , &t < \tau^*\ ,\\
                \frac{\kappa}{\gamma \sigma_S}\frac{\wt{Z}_t^{z^* + \theta}}{\wt{W}_t^{w^*}}\ ,& t \geq \tau^*\ ,
            \end{dcases}
        \end{equation}
        and
        \begin{equation}
            x^*_t = \begin{cases}
                b_0(\partial_w F(W_t^w), (P/\eta) \partial_w F(W_t^w) )\ , &t < \tau^*\ ,\\
                \ul{x} + (1 - \beta) \varphi \wt{Z}_t^{z^* + \theta}\ ,& t \geq \tau^*\ ,
            \end{cases}
        \end{equation}
        and finally
        \begin{equation}
            s^*_t =
            \begin{dcases}
                b_1(\partial_w F(W_t^w), (P/\eta) \partial_w F(W_t^w) )\ , &t < \tau^*\ ,\\
                \ul{s} + \beta \varphi \frac{\wt{Z}_t^{z^*+\theta}}{P/\wt{\eta}}\ ,& t \geq \tau^*\ ,
            \end{dcases}
        \end{equation}
        where $b_0$ and $b_1$ are deterministic functions defined in \eqref{eq:opt_control_formal_1} and \eqref{eq:opt_control_formal_2} respectively.
    \end{enumerate}
\end{thm}



We make a few remarks. Firstly, the above result underscores the centrality of the subsistence requirement $\ul{s}$ in this model: since the agent cannot reduce demand below this level, immediate investment in the retrofit is optimal if the subsistence net present value $\theta$ is non-negative. In fact, in this case the threshold satisfies $w^* = \Lambda \theta - H \le -H$, while admissibility requires $w > -H - \theta \ge -H$. Hence $w > w^*$ necessarily holds, so that the region $w < w^*$ cannot arise when $\theta \ge 0$, and immediate investment is the only feasible optimal policy. Investment is likewise immediate whenever $w \ge w^*$, in which case the agent invests at $t=0$ and attains the value $F(w) = G(w)$. On the other hand, when $\theta < 0$ and $w < w^*$, there is an option value of waiting to invest \citep[cf.][]{McDonald1986}. As regards allocation and consumption, the post-investment strategies have been discussed above, following Proposition~\ref{1d:prop:terminal_gain}. In contrast, in the waiting region the optimal strategies are available only in implicit form; however, they admit a natural closed-form approximation, which is introduced in the sequel.

\subsection{Approximate welfare results}
\label{1d:sec:welfare:approx}

Having thus solved the agent's decision problem, we consider now the social-cost implications developed in Section~\ref{1d:sec:welfare:definitions}. For ease of notation, the following results are stated directly in terms of disposable capital rather than wealth. We use the identity $w(z) \defeq z - H$ as well as the threshold $z^* = z(w^*) > 0$ for disposable capital in the sequel. We begin by stating the solution to the counterfactual decision problem \eqref{1d:eq:counterfactual_problem}, which follows immediately from Proposition~\ref{1d:prop:terminal_gain} by symmetry. 

\begin{cor}
\label{1d:cor:counterfactual}
    Let $z > 0$. The counterfactual value function in \eqref{1d:eq:counterfactual_problem} is given by
    \begin{equation}
        \wh{F}(z) = \wh{\Gamma}^{-\gamma} \frac{z^{1-\gamma}}{1-\gamma}\ ,
    \label{1d:eq:counterfac_val_func}
    \end{equation}
    where $\wh{\Gamma} > 0$ is a constant identical to $\Gamma$ of \eqref{eq:def_gamma} with $\wt{\eta}$ replaced by $\eta$. The optimal strategies are given by
    \begin{equation}
        \wh{a}_t^* = \frac{\kappa}{\gamma \sigma_S} \frac{\wh{Z}_t^{z}}{\wh{W}_t^{w(z)}}\ , \quad
        \wh{x}_t^* \defeq \ul{x} + (1 - \beta) \varphi \wh{Z}_t^{z}\ ,
        \quad
        \wh{s}_t^* \defeq \ul{s} + \frac{\beta \varphi \wh{Z}_t^{z}}{P / \eta}\ ,
    \label{1d:eq:approx_controls}
    \end{equation}
    where $\wh{Z}_t = z(\wh{W}_t)$ and $\varphi > 0$ is a constant given in \eqref{eq:def_varphi}. It follows that the dynamics of the optimally-controlled process $\wh{Z}_t$ are identical to those of $\wt{Z}_t$ from \eqref{1d:eq:opt_z_tilde_t}.
\end{cor}

\noindent Intuitively, the \enquote{never-invest} optimal strategies are wholly analogous to the \enquote{immediate-invest} optimal strategies, except with a different efficiency parameter.

Now, in order to facilitate closed-form results, we propose an approximation to the implicit controls in Theorem~\ref{1d:thm:val_func}. Recall that this is only necessary for $t < \tau^*$ in the case where waiting is optimal, i.e.~when $\theta < 0$ and $z < z^*$. A natural choice is to assume that the controls in this regime are approximated by the counterfactual controls from Corollary~\ref{1d:cor:counterfactual} above. That is, the approximation assumes that the agent allocates wealth and consumes as though they were never going to invest, but then investment does in fact occur when the threshold $z^*$ is attained. Appendix~\ref{append:base_model:approx_controls} presents a formal argument and a numerical example showing that the approximation introduces only small errors and, moreover, overstates energy consumption. Since welfare is decreasing in consumption in the present setting, this induces a downward bias in the value function. Consequently, the welfare results reported below are conservative. In other words, the results are close to exact, with both qualitative conclusions and economic intuition preserved.

\begin{approximation}
\label{1d:assump_approx}    
    Suppose $\theta < 0$ and $z \in (0, z^*)$. Then define the first hitting time
    \begin{equation}
        \wh{\tau}^*(z) \defeq \inf \{ t \geq 0 \mid \wh{Z}_t^z > z^* \}\ .
    \label{1d:eq:tau_star_approx}
    \end{equation}
    Consequently, approximate the optimal controls from Theorem~\ref{1d:thm:val_func} as
    \begin{equation}
        a_t \approx
        \begin{dcases}
           \frac{\kappa}{\gamma \sigma_S} \frac{\wh{Z}_t^{z}}{\wh{W}_t^{w(z)}}\ ,&t < \wh{\tau}^*\ ,\\
           \frac{\kappa}{\gamma \sigma_S}\frac{\wt{Z}_t^{z^* + \theta}}{\wt{W}_t^{w^*}}\ ,&t \geq \wh{\tau}^*\ ,
        \end{dcases}
    \end{equation}
    and
    \begin{equation}
        x_t \approx
        \begin{cases}
           \ul{x} + (1 - \beta) \varphi \wh{Z}_t^z\ ,&t < \wh{\tau}^*\ ,\\
           \ul{x} + (1 - \beta) \varphi \wt{Z}_t^{z^* + \theta}\ ,&t \geq \wh{\tau}^*\ ,
        \end{cases}
        \quad
        s_t \approx
        \begin{dcases}
           \ul{s} + \frac{\beta \varphi \wh{Z}_t^{z}}{P / \eta}\ ,&t < \wh{\tau}^*\ ,\\
           \ul{s} + \beta \varphi \frac{\wt{Z}_t^{z^*+\theta}}{P/\wt{\eta}}\ ,&t \geq \wh{\tau}^*\ .
        \end{dcases}
    \end{equation}
\end{approximation}

We immediately employ the above to examine the conditions under which rebound and backfire arise. We reiterate that the approximation is used only when $\theta < 0$ and $z < z^*$, ultimately to characterise the social cost in part (ii) of Theorem~\ref{1d:thm:welfare_improve} below; all remaining results are exact.

\begin{prop}
\label{1d:prop:rebound}
Let $z > 0$.
\begin{enumerate}
    \item Suppose $\theta > 0$. Then rebound occurs in expectation, i.e.~$\E[R_t] > 0$. The probability of backfire occurring at time $t$ is given by
    \begin{equation}
        \mathbb{P}(Q_t > 0) = 1 - \Phi \left( \frac{\log{\kappa} - \log{\theta} - (\mu_{\wt{Z}} - \tfrac{1}{2}\sigma_{\wt{Z}}^2) t}{\sigma_{\wt{Z}} \sqrt{t}} \right)\ ,
    \label{1d:eq:prob_backfire}
    \end{equation}
    where $\Phi$ is the cumulative distribution function of the standard normal distribution, and  $\kappa > 0$ is a constant defined in \eqref{1d:eq:def_kappa}. 
    
    \item Suppose $\theta = 0$. Then rebound occurs in expectation and backfire does not occur in expectation, i.e.~$\E[Q_t] < 0$.
    
    \item Suppose $\theta < 0$ and $z \geq z^*$. Then rebound occurs in expectation if
    \begin{equation}
        z > - \frac{\wt{\eta}\theta}{\wt{\eta} - \eta}\ .
    \label{1d:eq:rebound_pos_req}
    \end{equation}
    Backfire does not occur in expectation.
    \item Suppose $\theta < 0$ and $z < z^*$, and assume Approximation~\ref{1d:assump_approx}. Then rebound occurs in expectation if
    \begin{equation}
        \Lambda > - \frac{\wt{\eta}}{\wt{\eta} - \eta}\ ,
    \label{1d:eq:rebound_pos_req_2}
    \end{equation}
    where $\Lambda < 0$ is the constant from \eqref{1d:eq:trig_zstar}. Backfire does not occur in expectation.
\end{enumerate}
\end{prop} 

\noindent In sum, if the subsistence net present value $\theta > 0$, rebound always occurs in expectation; moreover, the probability of backfire occurring increases as time passes. Conversely, if $\theta < 0$, rebound occurs under what turns out to be relatively mild conditions (cf.~Section~\ref{1d:sec:cs:optimal_strategies}), whereas backfire never occurs. This general \enquote{no backfire} property for $\theta < 0$ proves important in the following result, the main one of this section, which quantifies the social implications of the energy-efficiency investment.



\begin{thm}
\label{1d:thm:welfare_improve}
    Let $z> 0$ and $\wh{\epsilon} - \mu_\varpi - \mu_{\wt{Z}} > 0$.
    \begin{enumerate}
        \item Suppose $\theta \geq 0$ or $z \geq z^*$. Then $V_{\mathrm{sc}}(z; \pi) = V_{\mathrm{sc}}(\pi) = \mathcal{I}(\pi)$, where
        \begin{equation}
            \mathcal{I}(\pi)
            = \left( \frac{(\wt{\eta}^{-1} - \eta^{-1}) \ul{s}}{\wh{\epsilon} - \mu_\varpi} + \frac{\beta \varphi \theta}{(\wh{\epsilon} - \mu_\varpi - \mu_{\wt{Z}}) P} \right) \pi\ .
        \label{1d:eq:social_cost_int}
        \end{equation}
        It follows that $V_{\mathrm{sc}} \leq 0$ if and only if
        \begin{equation}
            \theta \leq \frac{(\wh{\epsilon} - \mu_\varpi - \mu_{\wt{Z}}) (\wt{\eta} - \eta)  \ul{s} P}{(\wh{\epsilon} - \mu_\varpi) \beta \varphi \wt{\eta} \eta}\ .
        \label{1d:eq:theta_bnd}
        \end{equation}
        \item Suppose $\theta < 0$ and $z < z^*$, and assume Approximation~\ref{1d:assump_approx}. Then
        \begin{equation}
            V_{\mathrm{sc}}(z; \pi) = \mathscr{L}(z; \wh{\epsilon} - \mu_\varpi) \mathcal{I}(\pi)\ ,
        \label{1d:eq:social_cost_wait}
        \end{equation}
        where $\mathscr{L}(z; \varrho) = \E_z[\me^{-\varrho \wh{\tau}^*}]$ is the Laplace transform of $\wh{\tau}^*$, given explicitly in \eqref{1d:eq:laplace_transform}. It follows that $V_{\mathrm{sc}} \leq 0$ identically.
    \end{enumerate}
\end{thm}


\noindent The result shows that the retrofit is generally welfare improving, except in the limiting case of large $\theta$, corresponding to a relatively cheap retrofit, where early and prolonged backfire ultimately generates social costs (cf.~Case~(i) of Proposition~\ref{1d:prop:rebound}).



\subsection{Approximate optimal subsidy policy}
\label{1d:sec:subsidy:approx}

We consider now the design of the corrective subsidy for energy-efficiency adoption defined in \eqref{1d:eq:opt_subsidy_def}. Applying similar techniques as in Theorem~\ref{1d:thm:welfare_improve}, we are able to obtain an explicit solution in the case where the agent invests immediately, and to reduce to a deterministic optimisation if it is optimal for the agent to wait.



\begin{prop}
\label{1d:prop:approx_subsidy}
    Let $z > 0$.
    \leavevmode
    \begin{enumerate}
        \item Suppose $\theta \geq 0$ or $z \geq z^*$. The optimal subsidy rate in \eqref{1d:eq:opt_subsidy_def} is given by
        \begin{equation}
            m^*(\pi) = -\frac{C_1(\pi) + \xi_0 K}{\xi_1 K^2} < 0\ ,
        \end{equation}
        where $C_1 > 0$ is defined in \eqref{1d:eq:c_cons} below.
        \item Suppose $\theta < 0$ and $z < z^*$, and assume Approximation~\ref{1d:assump_approx}. The optimal subsidy rate in \eqref{1d:eq:opt_subsidy_def} is given by
        \begin{multline}
            m^*(z; \pi) = \argmin_{0 \leq m \leq \ol{m}} \bigg[ (D_0(z) + D_1(z) m)^{d_0} (C_0(\pi) + C_1(\pi) m)\\
            + (D_0(z) + D_1(z) m)^{d_1} \Psi(m K) \bigg]\ ,
        \label{1d:eq:sp_final}
        \end{multline}
        where $0 < \ol{m} < 1$ is defined in \eqref{1d:eq:ol_m_def}, and the functions $C_0$, $C_1$, $D_0$, $D_1$, and constants $d_0$ and $d_1$, are defined in the proof.
    \end{enumerate}
\end{prop} 

\noindent Hence, in light of Theorem~\ref{1d:thm:welfare_improve}, if $\theta \geq 0$ or $z \geq z^*$, a negative subsidy is justified in order to mitigate externalities due to backfire. However, if $\theta < 0$, since a social gain is guaranteed by Theorem~\ref{1d:thm:welfare_improve}, the subsidy level is positive to encourage earlier investment. 


\subsection{Aggregate behaviour}
\label{1d:sec:sol:aggregate}

We take up the aggregate quantities defined in Section~\ref{1d:sec:diffusion}. Since an explicit solution for the agent's optimal investment time is obtained in Theorem~\ref{1d:thm:val_func}, it is possible to simplify the expressions for the key quantities. Without risk of confusion, we denote the optimal investment time by $\tau$ in this section. Then with $\mathcal{F}_{\tau}$ the cumulative distribution function $\tau$, define
\begin{equation}
    \mathcal{G}(z^i; t) \defeq \mathbb{1}_{\{z^i \geq z^*\}} + \mathbb{1}_{\{z^i < z^*\}} \mathcal{F}_{\tau^i(z^i)}(t)\ ,
\end{equation}
which gives the probability that agent $i$ has invested at time $t$. Since the optimally-controlled state variable $Z_t$ is a geometric Brownian motion, an explicit expression for $\mathcal{F}_\tau$ is available \citep[Ch.~3.3.1]{Jeanblanc2009}, though we do not reproduce it here. It is then clear from \eqref{dc:eq:share} that the expected adoption share conditional on the population $\mathcal{P}$ is given by
\begin{equation}
    \E \left[ S_t \mid \mathcal{P} \right] = \frac{1}{N} \sum_{i=1}^N \mathcal{G}(z^i; t)\ ,
\label{1d:eq:expected_diffusion}
\end{equation}
which aggregates individual investment probabilities. Additionally, we note that the instantaneous rate of adoption, which is often of interest, is given by
\begin{equation}
    \frac{\diff}{\diff t} \E \left[ S_t \mid \mathcal{P} \right] = \frac{1}{N} \sum_{i = 1}^N \mathbb{1}_{\{z^i < z^*\}} f_{\tau^i(z^i)}(t)\ ,
\label{1d:eq:expected_rate}
\end{equation}
where $f_\tau$ is the probability density function of $\tau$. Similarly, expected aggregate energy consumption follows
\begin{equation}
    \E[C_t \mid \mathcal{P}] = \sum_{i = 1}^N \left[ (s_t / \eta) (1 - \mathcal{G}(z^i; t)) + (s_t / \wt{\eta}) \mathcal{G}(z^i; t) \right]\ .
\label{1d:eq:expected_tot_consump}
\end{equation}
Together, these quantities show how population-level trajectories are driven by the distribution of individual thresholds and the underlying financial uncertainty: agents above the wealth threshold adopt immediately, whereas those below it invest probabilistically over time, producing smooth aggregate dynamics.

\section{Case study}
\label{1d:sec:case_study}

The following sections present a detailed case study of an energy retrofit of a representative German single-family home, undertaken by an agent at the median of the wealth and income distributions. We analyse optimal strategies, welfare effects, optimal subsidy design, and aggregate quantities, concluding with comparative statics to assess parameter sensitivity. An excursion concerning optimal retrofit depth is presented in Appendix~\ref{append:base_model:depth}. 

The parameters required to specify the agent's decision problem are listed in Table~\ref{1d:tab:const_params}. The risk-free rate $\mu_R$, and the drift $\mu_S$ and volatility $\sigma_S$ for the risky asset are in line with standard values \citep[cf.][]{Kraft2011}. The labour income $Y$ as well as the initial condition for the wealth diffusion $w$ correspond to median values for a German homeowner with a mortgage \citep{Bundesbank2023}.\footnote{All monetary values are in \qty{2021}{\euro}.} The agent's relative risk aversion $\gamma$ and discount rate $\delta$ are also from \citet{Kraft2011}, with a hazard rate $\lambda$ corresponding to a remaining life expectancy of 50 years. The subsistence level of the numeraire $\ul{x}$ is taken to coincide with the tax-free basic allowance \citep{BMF2025}. The relative weighting $\beta$, subsistence level $\ul{s}$, as well as the price of gas $P$ were calibrated to obtain reasonable levels of energy-service demand. The efficiency parameters $\eta$, $\wt{\eta}$, borrowing rate $\rho$, and retrofit cost $K$ are estimated from \citet[Case Study \enquote{EFH78}]{Galvin2024}, corresponding to a typical German single-family home built during the period 1969--1978.

\begin{table}
    \caption{Parameter values for the case study in Section~\ref{1d:sec:case_study}. Sources in main text.}
    \footnotesize
    \begin{tabular}{c l l}
    \toprule
    Parameter      & Description             & Value \\
    \midrule
    \multicolumn{3}{c}{\emph{Financial assets}}\\ 
    $\mu_R$        & Drift, risk-free asset   & \qty{0.025}{\per\year} \\
    $\mu_S$        & Drift, risky asset       & \qty{0.07}{\per\year} \\
    $\sigma_S$     & Volatility, risky asset  & \qty{0.2}{\per\year} \\
    \midrule
    \multicolumn{3}{c}{\emph{Energy price, income, wealth}}\\ 
    $P$            & Gas price                & \qty{0.21}{\euro\per\kWh} \\
    $Y$            & Labour income            & \qty{47}{\kilo\euro\per\year} \\
    $w$            & Initial wealth           & \qty{45}{\kilo\euro} \\
    \midrule
    \multicolumn{3}{c}{\emph{Preferences}}\\ 
    $\beta$        & Weight, energy service   & 0.007 \\
    $\gamma$       & Relative risk aversion   & 4 \\
    $\delta$       & Discount rate            & \qty{0.03}{\per\year} \\
    $\lambda$      & Hazard rate              & \qty{0.02}{\per\year} \\
    \midrule
    \multicolumn{3}{c}{\emph{Subsistence consumption}}\\ 
    $\ul{x}$       & Non-energy good          & \qty{12}{\kilo\euro\per\year} \\
    $\ul{s}$       & Indoor temperature       & \qty{15}{\celsius} \\
    \midrule
    \multicolumn{3}{c}{\emph{Retrofit parameters}}\\
    $A$            & Dwelling area            & \qty{157}{\meter\squared} \\
    $\eta$         & Efficiency, existing state & \qty{0.005}{\celsius\per\watt} \\
    $\wt{\eta}$    & Efficiency, post-retrofit  & \qty{0.025}{\celsius\per\watt} \\
    $\rho$         & Borrowing rate           & \qty{0.04}{\per\year} \\
    $K$            & Retrofit cost            & \qty{120}{\kilo\euro} \\
    \bottomrule      
    \end{tabular}
\label{1d:tab:const_params}
\end{table}

\subsection{Optimal strategies}
\label{1d:sec:cs:optimal_strategies}

Firstly, due to the large cost of the retrofit, the subsistence net present value of the project is negative, namely $\theta = \qty{-16.2}{\kilo\euro}$. The corresponding investment threshold $w^* = \qty{430}{\kilo\euro}$, roughly 10 times the initial wealth level $w = \qty{45}{\kilo\euro}$. Figure~\ref{fig:1d_dist_tau} shows the probability density and cumulative distribution functions of the investment time $\tau^*$ for this initial wealth level, along with scaled multiples for comparison. As expected, the higher the initial wealth level, the earlier the expected investment time, and the higher the cumulative probability of investment.

The optimal allocation and consumption strategies from Theorem~\ref{1d:thm:val_func} are depicted in Figure~\ref{fig:1d_opt_controls} in a neighbourhood of the investment threshold. Comparing the pre- and post-investment regimes, a few patterns emerge. Allocation $a$ closely aligns with the Merton portfolio allocation, with the agent engaging in limited borrowing against future labour income to increase exposure to the risky asset. Moreover, allocation decreases slightly following the investment, since capital is tied up in the illiquid retrofit investment. Non-energy consumption $x$ also decreases after investment, since more of the budget is allocated to the energy service $s$, which is obtained with greater efficiency. Demand for the energy service $s$ increases sharply following investment, whereas the corresponding fuel consumption $c = s / \eta$ (resp.~$c = s /\wt{\eta}$ after investment) experiences a sharp decline due to the increased efficiency. That is, we have rebound without backfire, as established by Proposition~\ref{1d:prop:rebound}.


\begin{figure}
    \centering
    \includegraphics[width=\linewidth]{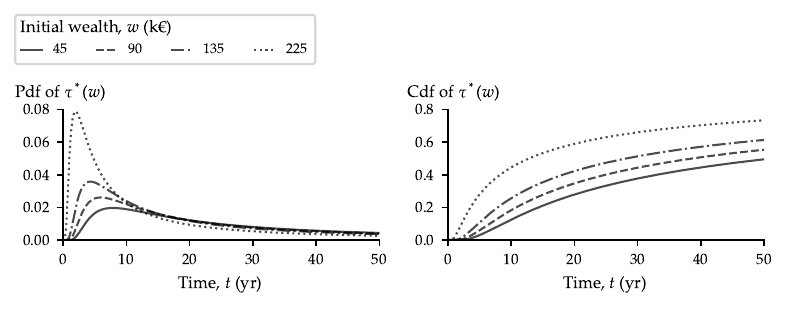}
    \caption{Probability density function (Pdf) and cumulative distribution function (Cdf) of the investment time $\tau^*$ for multiples of $w$ from Table~\ref{1d:tab:const_params}.}
    \label{fig:1d_dist_tau}
\end{figure}

\begin{figure}
    \centering
    \includegraphics[width=\linewidth]{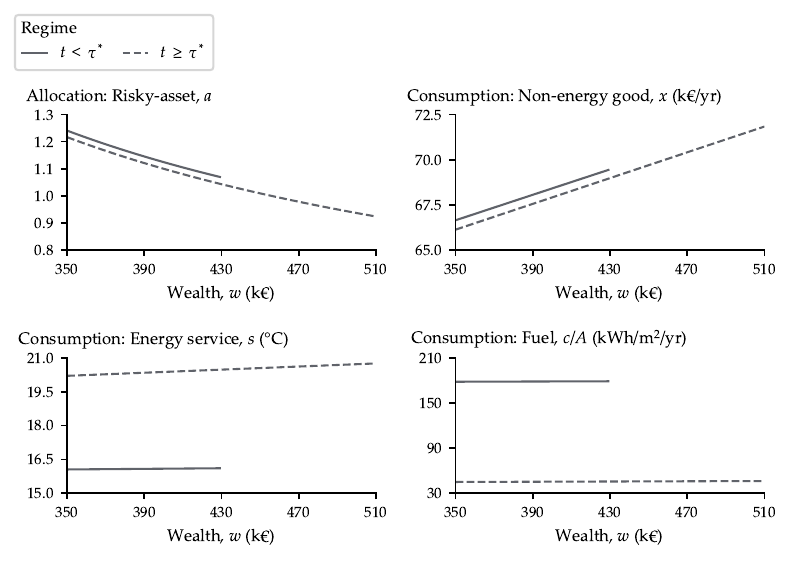}
    \caption{Optimal strategies from Theorem~\ref{1d:thm:val_func} for the pre- and post-investment regimes in a neighbourhood of the investment threshold $w^* = \qty{430}{\kilo\euro}$. For readability, fuel consumption $c$ is normalised to the dwelling area $A$. Since the agent invests as soon as $w > w^*$, the regime $t < \tau^*$ cuts off at this point. On the other hand, since the agent's wealth may fall below the threshold after investment, the domain of the controls for the regime $t \geq \tau^*$ is $w \in [-\wt{H}, \infty)$.}
    \label{fig:1d_opt_controls}
\end{figure}

\begin{figure}
    \centering
    \includegraphics[width=\linewidth]{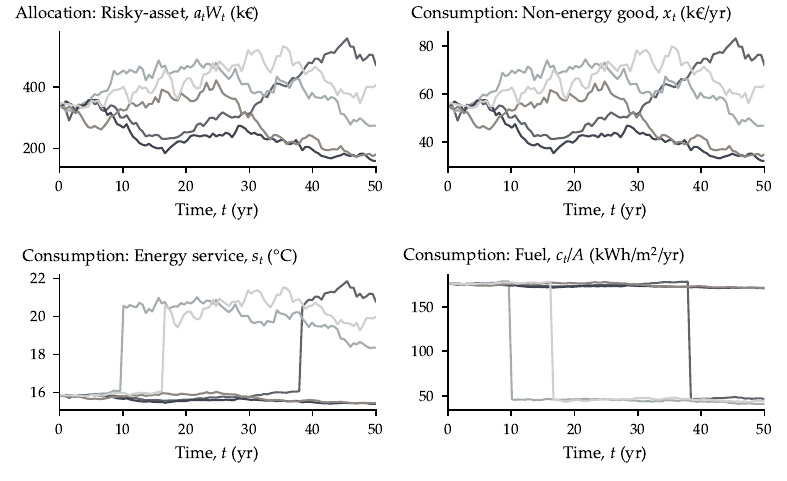}
    \caption{Optimal strategies from Theorem~\ref{1d:thm:val_func} for five exemplary trajectories. Investment takes place in three of the trajectories, as indicated by the demand jumps in the bottom row.  Fuel consumption $c_t$ is normalised to the dwelling area $A$.}
    \label{fig:1d_sim}
\end{figure}
 
The above observations can be appreciated in a different light in Figure~\ref{fig:1d_sim}, which depicts five optimally-controlled trajectories simulated over a \num{50} year horizon.\footnote{We display trajectories instead of expectations since the demand jumps due to investments are masked completely if expectations are shown.} The upwards jumps in energy-service demand clearly depict the moment of investment, with the corresponding downward jumps in the facing panel demonstrating how fuel consumption decreases for these same trajectories. The rebound and backfire measures $R_t$ and $Q_t$ defined in Section~\ref{1d:sec:welfare:definitions} quantify this change directly, and are shown in the bottom row of Figure~\ref{fig:1d_rebound_all}. For completeness, the top row of the same graphic presents the analogues of these measures, i.e.~the difference to the counterfactual controls of \eqref{1d:eq:counterfac_val_func}, for the allocation and non-energy consumption strategies. One sees clearly the effects discussed above: exposure to the risky asset increases prior to investment in order to build up capital stock, but decreases afterward as wealth becomes tied up in the retrofit project. On the other hand, consumption of the non-energy good declines relative to the counterfactual in trajectories where investment occurs, since the agent allocates relatively more resources to the efficient energy service. 

\begin{figure}
    \centering
    \includegraphics[width=\linewidth]{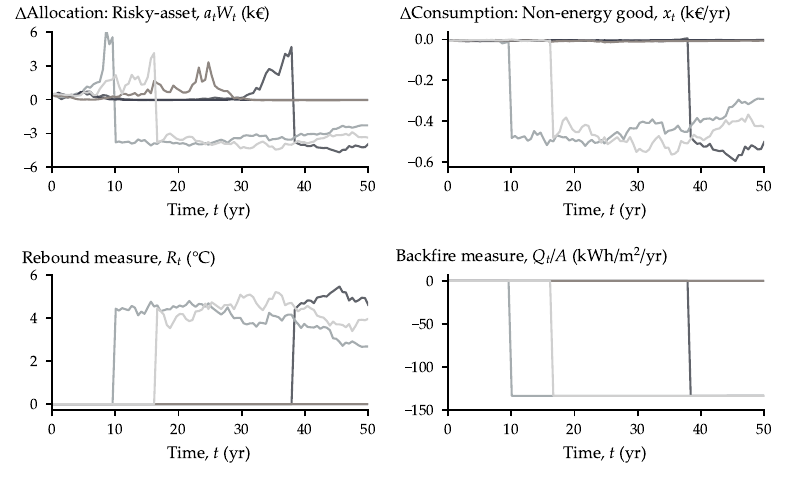}
    \caption{Difference between the optimal strategies from Theorem~\ref{1d:thm:val_func} and the counterfactual strategies from Corollary~\ref{1d:cor:counterfactual} for the five trajectories shown in Figure~\ref{fig:1d_sim}. The backfire measure $Q_t$ is normalised to the dwelling area $A$.}
    \label{fig:1d_rebound_all}
\end{figure}

\subsection{Social cost \& optimal subsidy policy}
\label{1d:sec:cs:welfare}

We move on to the social implications of the agent's decisions. The additional parameters required for the analysis in Section~\ref{1d:sec:welfare:approx} are listed in Table~\ref{1d:tab:planner_params}. The social discount rate $\epsilon$ is chosen slightly lower than the risk free rate $\mu_R$ as per \citet{Caplin2004}, the drift in social cost $\mu_\varpi$ is estimated from the long term carbon price scenarios in \citet{Gerlagh2017}, and the marginal cost of public funds $\xi_0$ is from \citet{Kleven2003}.\footnote{The friction parameter $\xi_1$ was fixed by trial and error; see Section~\ref{1d:sec:cs:cs} following for the comparative statics.}

Based on these parameter choices, Figure~\ref{fig:1d_combined} shows the measure of social benefit $-V_\mathrm{sc}$ from Theorem~\ref{1d:thm:welfare_improve}, computed over a grid of carbon prices and initial wealth.\footnote{Note that the social cost of energy consumption $\varpi_t$ equals the product of the carbon price and the emissions factor for the fuel, in this case gas, which is \qty{0.240e-3}{\tC\per\kWh} \citep{Koffi2017}. The range of carbon prices considered in Figure~\ref{fig:1d_combined} corresponds roughly to the scenarios in the supplementary material of \citet{Gerlagh2017}, normalised to \qty{2021}{\euro}.} As expected, since $\theta < 0$, welfare change is positive everywhere and increasing in carbon prices and wealth. It is also apparent that $V_\mathrm{sc}(w)$ is constant for $w \geq w^*$, as established by Theorem~\ref{1d:thm:welfare_improve}. The optimal subsidy $m^*$ from Proposition~\ref{1d:prop:approx_subsidy} is also shown in Figure~\ref{fig:1d_combined}; it is seen to be sharply increasing in carbon prices and slowly decreasing in wealth. The overall level of the subsidy is rather modest, attaining a maximum value of around \qty{1.5}{\percent} of the total retrofit cost on the considered grid. For wealth levels exceeding $w^*$, the penalty rate from Proposition~\ref{1d:prop:approx_subsidy}, which is anyway constant in $w$, is also roughly constant over the carbon price range \qtyrange[range-units=single]{10}{70}{\euro\per\tC}. It is given by $m^* = \qty{-1.76}{\percent}$ (not shown in Figure~\ref{fig:1d_combined}).



\begin{table}[t]
    \footnotesize
    \caption{Additional parameter values for the case study in Section~\ref{1d:sec:case_study}. Sources in main text.}
    \begin{tabular}{c l l}
    \toprule
    Parameter      & Description             & Value \\
    \midrule
    $\epsilon$     & Discount rate, social planner       & \qty{0.02}{\per\year} \\
    $\mu_\varpi$   & Drift, marginal social cost        & \qty{0.013}{\per\year} \\
    $\xi_0$        & Marginal cost of public funds      & 2.12 \\
    $\xi_1$        & Friction parameter, public funds  & \qty{1}{\per\euro} \\
    \bottomrule      
    \end{tabular}
\label{1d:tab:planner_params}
\end{table}

\begin{figure}
    \centering
    \begin{subfigure}{0.49\linewidth}
        \centering
        \includegraphics[width=\linewidth, trim={0.25cm 0.45cm 0.25cm 0.95cm}, clip]{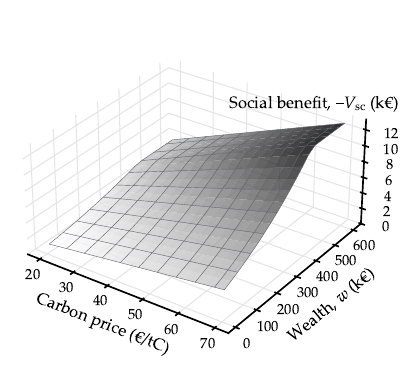}
    \end{subfigure}
    \hfill
    \begin{subfigure}{0.49\linewidth}
        \centering
        \includegraphics[width=\linewidth, trim={0.25cm 0.45cm 0.25cm 0.95cm}, clip]{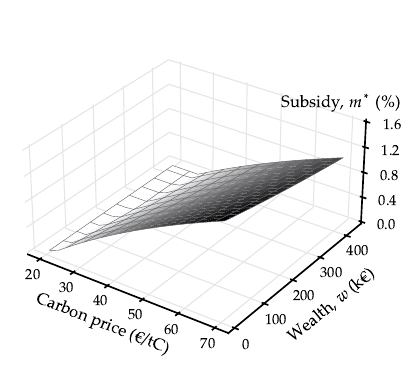}
    \end{subfigure}
    \caption{On the left, the social benefit $-V_\mathrm{sc}$ from Theorem~\ref{1d:thm:welfare_improve}; and on the right, the optimal subsidy level $m^*$ from Proposition~\ref{1d:prop:approx_subsidy}, each computed over a grid of carbon prices and wealth levels.}
    \label{fig:1d_combined}
\end{figure}


\subsection{Aggregate behaviour}
\label{1d:cs:diffusion}

This section demonstrates the aggregate quantities of interest discussed in Section~\ref{1d:sec:diffusion}. Firstly, we restrict the physical scope of the study as follows: the single family home studied above is representative of a large cohort of the German building stock, approximately \num{1.5} million dwellings \citep{Loga2015}; assuming that around half of these are homeowners \citep[cf.][]{Destatis2025}, we arrive at approximately \num{750000} dwellings. The risk-free rate $\mu_R$, the risky-asset parameters $\mu_S$ and $\sigma_S$, and energy price $P$ are assumed common to all agents. For simplicity, the following calculation also ignores variations in the borrowing rate $\rho$, retrofit cost $K$, and the efficiency parameters $\eta$ and $\wt{\eta}$, keeping them fixed at the respective levels in Table~\ref{1d:tab:const_params}. The remaining parameters are idiosyncratic. We have the preference parameters $\beta$, $\gamma$, $\delta$, and $\lambda$, and the subsistence consumption levels $\ul{x}$ and $\ul{s}$. Independent uniform distributions are assumed for these parameters, with the distributions taken to be centred around the values in Table~\ref{1d:tab:const_params} with a width of \qty{\pm10}{\percent}. Then, a joint distribution for labour income $Y$ and initial wealth $w$ was calibrated to data from the \citet{Bundesbank2023}.\footnote{We summarise the procedure briefly. Median and mean values for German homeowner wealth and income determined log-normal marginals for $Y$ and $w$. The empirical profile of mean wealth across income quantiles was then used to estimate a log--log relation $\log \mathbb{E}[w \mid Y] = \alpha + \beta \log Y$, which allowed us to recover the conditional variance via the law of total variance. For sampling, we draw $\log{Y}$ from its marginal distribution and $\log{W}$ conditionally on $\log{Y}$, yielding cross-sectional samples consistent with the reported moments.} By drawing from these assumed distributions, we generate a population of representative agents, as in \eqref{1d:eq:population}.

Figure~\ref{fig:dc_diffusion} shows the expected share of adoption and expected total energy consumption, along with other possible trajectories. We see that immediate investment is optimal for roughly half of adopters, with uptake being rather gradual for the remaining share. Consequently, total energy consumption falls slowly over time as the share of adopters increases. Nevertheless, significant variation among the trajectories is observed. On the other hand, Figure~\ref{fig:dc_subsidy} demonstrates the effect of the subsidy policy from Proposition~\ref{1d:prop:approx_subsidy} on adoption share and energy consumption; for realism, we ignore the penalty for immediate investment, i.e.~Case (i) of Proposition~\ref{1d:prop:approx_subsidy}, retaining only the positive subsidies which encourage investment, i.e. Case (ii). It is apparent that the main effect of the policy is to increase the share of agents who adopt immediately, thus shifting the entire curve upward. The shift in the rate of adoption itself is slight.

\begin{figure}
    \centering
    \includegraphics[width=\linewidth]{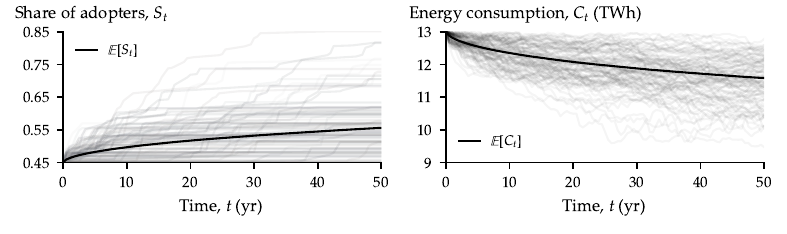}
    \caption{On the left: the expected share of adopters $\E[S_t]$ from \eqref{1d:eq:expected_diffusion}, shown with other simulated trajectories. On the right, expected aggregate energy consumption $\E[C_t]$ from \eqref{1d:eq:expected_tot_consump} for the same set of trajectories.}
    \label{fig:dc_diffusion}
\end{figure}

\begin{figure}
    \centering
    \includegraphics[width=\linewidth]{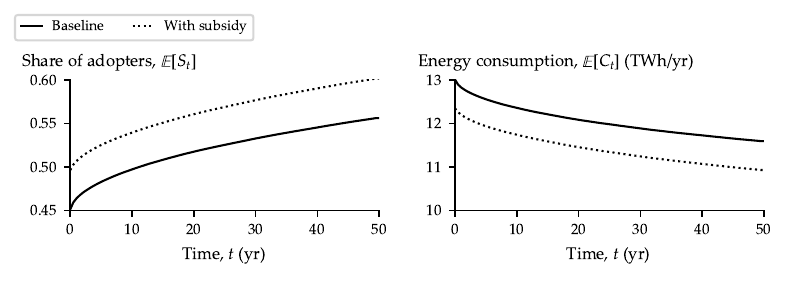}
    \caption{On the left: the expected share of adopters $\E[S_t]$ for the baseline scenario from Figure~\ref{fig:dc_diffusion} together with a second scenario where the optimal subsidy policy from Case (ii) of Proposition~\ref{1d:prop:approx_subsidy} is implemented. The panel on the right shows the analogous quantities for expected aggregate energy consumption $\E[C_t]$.}
    \label{fig:dc_subsidy}
\end{figure}

\subsection{Comparative statics}
\label{1d:sec:cs:cs}

We take up the comparative statics of the case study. Since the model is complex, involving 15 parameters and several interdependencies, the following analysis will not be exhaustive. Instead, we focus firstly on the local effects of the model parameters on two fundamental outputs: the investment threshold $w^*$ and optimal subsidy level $m^*$. Secondly, we examine the effects of three key model parameters on the optimal strategies. Finally, we study the effect of market volatility on technology diffusion.

Table~\ref{1d:tab:cs:cs} lists the elasticities at $w^*$ and $m^*$ of the baseline parameters from Tables~\ref{1d:tab:const_params} and~\ref{1d:tab:planner_params}, assuming a carbon price of \qty{45}{\euro\per\tC}.\footnote{Since a closed-form expression for $w^*$ is available, the elasticities are exact; the numerical calculations employed automatic differentiation \citep{Maclaurin2015}. On the other hand, since $m^*$ does not admit an explicit solution, a finite-difference scheme was implemented.} Consider first the elasticities at $w^*$. Several are extremely large, indicating a certain level of model misspecification. These include the retrofit parameters $\eta$, $\wt{\eta}$, $\rho$, and $K$, the energy price $P$, and the subsistence level $\ul{s}$.\footnote{The elasticities of $\rho$ and $K$ are identical since only the combination $\rho K$ is present in the expression for $w^*$; the same is true for the combination $\ul{s} P$.} Each of these has an outsize effect on the retrofit threshold, and by extension on the other optimal strategies, indicating that an upper bound on energy-service consumption within the utility function would improve model realism. On the other hand, labour income $Y$ and the preference weight $\beta$ display large but plausible influences on $w^*$. The financial market parameters exhibit large to medium elasticities, with an increase in the risk-free rate $\mu_R$ increasing the attractiveness of the investment; conversely, if the risky asset becomes more attractive (higher $\mu_S$, lower $\sigma_S$) the retrofit investment becomes relatively less attractive. An increase in the subsistence level of the non-energy good $\ul{x}$ is seen to drive up the investment threshold, since comparatively more resources must be allocated to basic non-energy consumption. The remaining parameters manifest only a small local influence on $w^*$.

The optimal subsidy rate $m^*$ is seen to be most sensitive to parameters that directly affect financing conditions and the effective return on energy efficiency. As above, the retrofit parameters, together with the gas price $P$ and subsistence level $\ul{s}$ exhibit the largest effects; moreover, the direction of the effects is identical to the previous case, which is intuitive. For instance, a higher baseline efficiency $\eta$ decreases the attractiveness of the retrofit for the agent, so the planner compensates by increasing $m^*$. As regards market conditions, the parameters $\mu_R$, $\mu_S$ and $\sigma_S$ exhibit moderate-to-large effects in the \emph{opposite} direction as the effects on $w^*$. Hence, as financial markets become more rewarding, both the agent and the planner withdraw support for retrofit investment; the former because it is privately beneficial to do so, the latter because it is fiscally less efficient to subsidise. Conversely, as markets become riskier, or as the risk-free rate increases, both the agent and the planner shift toward the safer, socially productive retrofit. The subsidy level is seen to be progressive, with higher income and wealth levels associated with lower levels of subsidy, although the effect of the wealth level is comparatively small. Finally, among the social planner's parameters, the largest effects are attributed to the initial social cost $\pi$ and marginal cost of public funds $\xi_0$.

\begin{table}
    \footnotesize
    \centering
    \caption{Local parameter elasticities for the case study in Section~\ref{1d:sec:case_study}.}
    \begin{tabular}{clrr}
    \toprule
    Parameter & Description & Elasticity~$w^*$ & Elasticity~$m^*$ \\
    \midrule
    \multicolumn{4}{c}{\emph{Financial assets}}\\
    $\mu_R$ & Drift, risk-free asset & -1.15 & 0.75 \\
    $\mu_S$ & Drift, risky asset & 0.67 & -3.87 \\
    $\sigma_S$ & Volatility, risky asset & -0.43 & 2.45 \\
    \midrule
    \multicolumn{4}{c}{\emph{Energy price, income, wealth}}\\
    $P$ & Gas price & -40.20 & -4.12 \\
    $Y$ & Labour income & -4.37 & -2.29 \\
    $w$ & Initial wealth &  & -0.06 \\
    \midrule
    \multicolumn{4}{c}{\emph{Agent preferences}}\\
    $\beta$ & Weight, energy service & -3.72 & -1.49 \\
    $\gamma$ & Relative risk aversion & -0.44 & 1.25 \\
    $\delta$ & Discount rate, agent & -0.18 & 0.65 \\
    $\lambda$ & Hazard rate & -0.12 & -0.39 \\
    $\ul{x}$ & Subsistence level, numeraire & 1.12 & 0.59 \\
    $\ul{s}$ & Subsistence level, energy service & -40.20 & -2.54 \\
    \midrule
    \multicolumn{4}{c}{\emph{Retrofit parameters}}\\
    $\eta$ & Efficiency, existing state & 52.69 & 3.10 \\
    $\wt{\eta}$ & Efficiency, post-retrofit & -12.49 & -1.53 \\
    $\rho$ & Borrowing rate & 44.46 & 5.35 \\
    $K$ & Retrofit cost & 44.46 & 3.07 \\
    \midrule
    \multicolumn{4}{c}{\emph{Social planner parameters}}\\
    $\epsilon$ & Discount rate, planner &  & -0.82 \\
    $\pi$ & Initial social cost &  & 1.68 \\
    $\mu_\pi$ & Drift, social cost &  & 0.70 \\
    $\xi_0$ & Marginal cost of public funds &  & -1.23 \\
    $\xi_1$ & Friction parameter, public funds &  & -0.45 \\
    \bottomrule
    \end{tabular}
    \label{1d:tab:cs:cs}
\end{table}

Next, we examine the effects of changing three key parameters, namely the preference weight $\beta$, risk aversion $\gamma$, and risky-asset volatility $\sigma_S$, on the agent's allocation and consumption optimal strategies. For ease of comparison, we restrict attention to the regime $t > \tau^*$, since the controls are defined over the entire wealth domain (cf.~Figure~\ref{fig:1d_opt_controls}). The parameters were varied by $\qty{\pm10}{\percent}$ relative to their baseline values in Table~\ref{1d:tab:cs:cs}, with the resulting controls displayed in Figure~\ref{fig:1d_controls_cs}.\footnote{For $\sigma_S$, the $+\qty{10}{\percent}$ variation results in a mild violation of the patience condition \eqref{1d:eq:patience_cond}.} The findings are intuitive and consistent with expectations. The preference weight $\beta$ has almost no effect on either portfolio allocation or non-energy consumption; however, its influence on energy-service consumption is strong and in the expected direction, with higher values associated with higher energy-service demand.  The risk-aversion parameter $\gamma$ significantly affects each of the optimal strategies: higher risk aversion is associated with lower investment in the risky asset and reduced consumption of both goods. The same holds for the risky-asset volatility $\sigma_S$, which influences portfolio allocation even more strongly than $\gamma$. The effects on consumption are also clear: greater market volatility leads to lower consumption as the agent responds through increased precautionary saving.

\begin{figure}
    \centering
    \includegraphics[width=\linewidth]{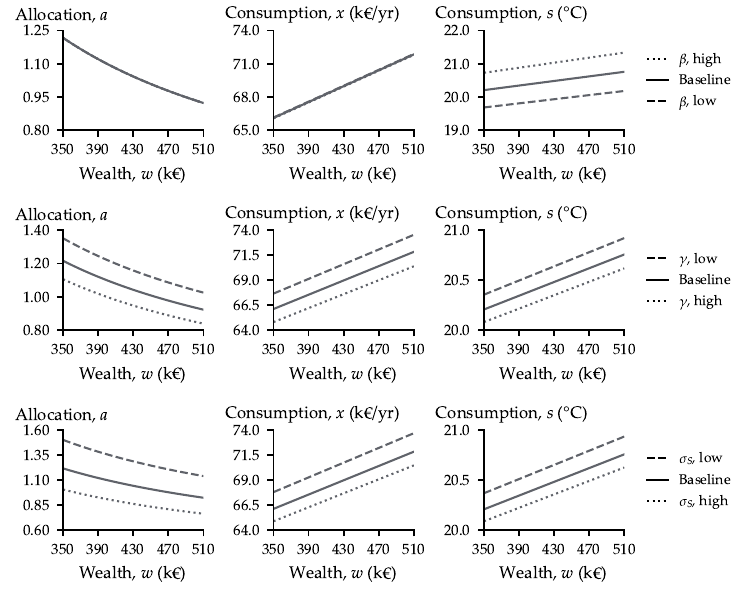}
    \caption{Comparative statics for the model parameters $\beta$, $\gamma$, and $\sigma_S$. The baseline controls are identical to Figure~\ref{fig:1d_opt_controls}, restricted to the regime $t > \tau^*$. The \enquote{low} and \enquote{high} values for the parameters correspond to $\pm \qty{10}{\percent}$ changes relative to the values in Table~\ref{1d:tab:cs:cs}.}
    \label{fig:1d_controls_cs}
\end{figure}

Finally, we examine the effect of market volatility on aggregate technology adoption and energy consumption. With the same representative agents used for the simulations in Section~\ref{1d:cs:diffusion}, we examine two volatility scenarios, namely \qty{\pm 10}{\percent} of the value in Table~\ref{1d:tab:const_params}. Figure~\ref{fig:dc_comp_stat} depicts the effects of these scenarios on expected adoption share $\E[S_t]$ and expected aggregate energy consumption $\E[C_t]$. An interesting pattern emerges: although higher volatility lowers each agent’s investment threshold (cf.~Table~\ref{1d:tab:cs:cs}) and slightly increases the initial share of early adopters relative to the baseline, the long-run effect on cumulative adoption moves in the \emph{opposite} direction, namely, high volatility ultimately reduces adoption relative to the baseline, whereas low volatility leads to higher cumulative uptake. The underlying mechanism is as follows: higher volatility widens the dispersion of wealth paths, generating a small group of very early adopters, but an even larger group of agents whose wealth remains persistently below the investment region and therefore fails to reach the threshold within the finite horizon. Consequently, the aggregate adoption share, which depends on the entire distribution of stopping times rather than the marginal shift in individual thresholds, is systematically reduced.

\begin{figure}
    \centering
    \includegraphics[width=\linewidth]{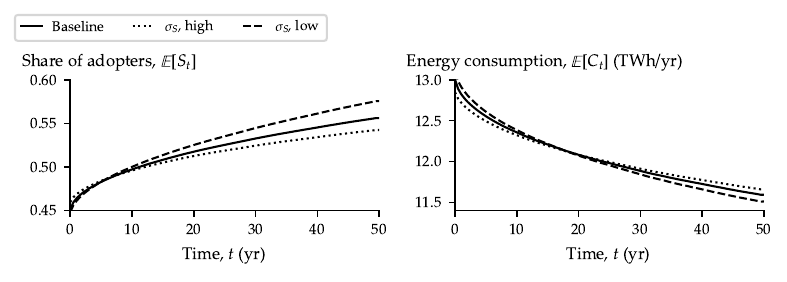}
    \caption{Effect of the risky-asset volatility $\sigma_S$ on expected adoption share $\E[S_t]$ and expected total energy consumption $\E[C_t]$; the baselines are from Figure~\ref{fig:dc_diffusion}.}
    \label{fig:dc_comp_stat}
\end{figure}

Figure~\ref{fig:dc_sub_comp_stat} repeats this exercise, but now with the subsidy included (cf.~Figure~\ref{fig:dc_subsidy}). Comparing the left panels of Figures~\ref{fig:dc_comp_stat} and~\ref{fig:dc_sub_comp_stat}, we see that the subsidy indeed counteracts the wealth dispersion from increased market volatility by compressing the distribution of investment thresholds; this stabilises the diffusion paths and reduces the long-term spread between the volatility scenarios. This makes sense in light of the policy's aim, which is to mitigate externalities from energy consumption, not simply to drive aggregate adoption; as such, the policy effectively targets energy-efficiency adoption where it generates the greatest marginal benefit.

\begin{figure}
    \centering
    \includegraphics[width=\linewidth]{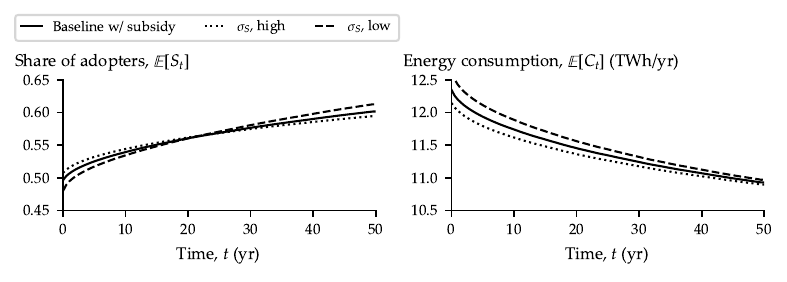}
    \caption{Effect of the risky-asset volatility $\sigma_S$ on expected adoption share $\E[S_t]$ and expected total energy consumption $\E[C_t]$ in the presence of the optimal subsidy from Case (ii) of Proposition~\ref{1d:prop:approx_subsidy}.; the baselines are from Figure~\ref{fig:dc_subsidy}.}
    \label{fig:dc_sub_comp_stat}
\end{figure}

\section{Conclusions \& outlook}
\label{1d:sec:discuss}

This article developed a model of consumption, investment, and energy-efficiency technology adoption under uncertainty. Despite its stylised nature, the model yielded substantial quantitative and qualitative insights, many in closed form. It demonstrated that the agent’s optimal strategies, particularly the adoption of the energy-efficiency technology, are directly contingent on wealth, and that investment timing, energy demand, and portfolio choices co-evolve within a unified intertemporal structure. This joint evolution provided dynamic, internally consistent definitions of rebound and backfire effects, social cost, and optimal subsidy design, with several components characterised in closed form or through tractable approximations.

The analysis further showed that agent heterogeneity plays a central role: differences in wealth, preferences, income, and technology characteristics generate pronounced variation in adoption incentives, welfare outcomes, and responses to financial-market conditions. These heterogeneous responses aggregate in non-linear ways, such that macro-level energy use and technology diffusion emerge endogenously from micro-level optimisation. Comparative statics demonstrated how changes in energy prices, volatility, and financial parameters shift investment thresholds and alter aggregate adoption patterns.

These features have direct policy implications. Because investment incentives and welfare effects vary systematically across agents, effective subsidy design must account for heterogeneity rather than rely on uniform support. Moreover, since adoption thresholds respond sensitively to external conditions, particularly energy prices and macroeconomic volatility, subsidy schemes require regular recalibration to track underlying incentives and limit free-riding. The framework therefore provides a disciplined basis for identifying where subsidies have the highest marginal effect and for assessing how policies interact with non-linear aggregation effects that influence system-wide adoption.


The outlook for this research involves relaxing the model's more idealised assumptions. Since these simplifications were introduced to emphasise analytical tractability, it is likely that extensions will rely heavily on numerical methods. Three avenues appear particularly consequential. First, as the realism of theoretical and numerical results was dampened by the potential for unbounded energy consumption, constraining agent preferences is expected to improve alignment with observed behaviour. This connects to a broader research gap concerning the specification of utility over energy services and the empirical identification of rebound effects and price elasticities. The development of preference models informed by micro-data, field experiments, or targeted elicitation of discounting and risk attitudes would substantially enhance the foundations of the framework. Secondly, as human capital was shown to play a central role, incorporating uncertainty in labour income or introducing retirement is likely to generate quantitatively different outcomes. Finally, as energy prices were found to exert a decisive influence on both consumption and investment decisions, accounting for price dynamics is essential for a comprehensive understanding of agent behaviour.

As for extensions beyond partial equilibrium, a natural direction concerns the evaluation of optimal corrective taxation for emissions. A tractable next step is to examine how an energy tax, coupled with a revenue-recycling scheme, affects energy consumption and technology adoption when households differ in wealth and thus in their sensitivity to operating costs. Because adoption in the present model is explicitly wealth-dependent, a uniform tax raises operating expenses unevenly across agents, potentially suppressing uptake among liquidity-constrained households even when the tax is welfare-improving on environmental grounds. Recycling revenues as lump-sum transfers or targeted subsidies can counteract this channel, but the distributional and efficiency implications depend on how transfers are allocated. A focused analysis of this tax–transfer mechanism would permit a clean quantification of how redistribution interacts with adoption dynamics and whether appropriately designed recycling can preserve the corrective intent of the tax while mitigating adverse effects on diffusion. Further general-equilibrium considerations of interest include how the tax–transfer mechanism interacts with endogenous wages, savings behaviour, and equilibrium energy prices, as well as potential feedbacks through capital accumulation or sectoral energy supply, thereby shaping both the aggregate diffusion path and the welfare distribution across heterogeneous households

\section*{Acknowledgements}

AB gratefully acknowledges support from the Karlsruhe House of Young Scientists through a Research Travel Grant, during which part of this research was carried out at the Norwegian University of Science and Technology (NTNU). We thank Wolf Fichtner, Stein-Erik Fleten, Nathan Chan, Emil Kraft, Verena Hagspiel, and Maria Lavrutich for helpful discussions and feedback. We also thank participants at the Magdeburg Workshop for Academics in Real Options 2023, particularly Stefan Kupfer and Benoît Chevalier-Roignant, for helpful comments.

\section*{Competing interests}

The authors declare no competing interests.

\begin{appendices}

\setcounter{equation}{0}
\renewcommand{\theequation}{\thesection\arabic{equation}}

\section{Mathematical proofs}
\label{append:base_model:proofs}


\begin{proof}[Proof of Proposition~\ref{1d:prop:terminal_gain}]
    The proof is standard. Firstly, the dynamic programming equation associated to \eqref{1d:eq:def_g} is
    \begin{equation}
        \wh{\delta} G - \sup_{\wt{a}, \wt{x}, \wt{s}} \left[ \mathcal{L}^{\wt{a}, \wt{x}, \wt{s}}_{w} G + U(\wt{x}, \wt{s}) \right] = 0\ ,
    \label{eq:hjb_g}
    \end{equation}
    where 
    \begin{equation}
        \mathcal{L}^{\wt{a}, \wt{x}, \wt{s}}_{w} G =  ( \wt{a} \kappa \sigma_S w + \mu_R w + \wt{Y} - \wt{x} - (\wt{s} / \wt{\eta}) P ) \partial_{w} G + \frac{1}{2} \wt{a}^2 \sigma_S^2 w^2 \partial_{w}^2 G\ .
    \label{eq:lagrange_g}
    \end{equation}
    Assume $\partial_{w} G > 0$ and $\partial_{w}^2 G < 0$ and perform the optimisations over the controls to reduce \eqref{eq:hjb_g} to the partial differential equation
    \begin{equation}
        \wh{\delta} G - \left( (\mu_R w + \wt{Y}) \partial_{w} G - \frac{\kappa^2 (\partial_{w} G)^2}{2 \sigma_S^2 \partial_{w}^2 G} + \wh{U}(\partial_{w} G, (P / \wt{\eta}) \partial_{w} G ) \right) = 0\ ,
    \label{eq:hjb_g_ustar}
    \end{equation}
    where
    \begin{equation}
        \wt{a}^* = -\frac{\kappa}{\sigma_S}\frac{\partial_{w} G}{w \partial_{w}^2 G}\ ,
    \label{eq:opt_invest_formal}
    \end{equation}
    and $\wh{U}$ is the Legendre-Fenchel transform of $U$, defined as
    \begin{align}
        \wh{U}(\pi, \xi) &= \sup_{b_0, b_1} \left[ U(b_0, b_1) - b_0 \pi - b_1\xi  \right]\\
        &= - \ul{x} \pi - \ul{s} \xi + \frac{1}{\wh{\gamma}} \left( \left( \frac{\pi}{1- \beta} \right)^{1-\beta} \left( \frac{\xi}{\beta} \right)^{\beta}  \right)^{-\wh{\gamma}}
    \label{eq:leg_transform}
    \end{align}
    with $\wh{\gamma} \defeq (1-\gamma)/\gamma$ and corresponding optimal controls
    \begin{align}
        b_0(\pi, \xi) &= \ul{x} + \left( \frac{\pi}{1-\beta} \right)^{-1} \left( \left( \frac{\pi}{1- \beta} \right)^{1-\beta} \left( \frac{\xi}{\beta} \right)^{\beta}  \right)^{-\wh{\gamma}}\ ,
        \label{eq:opt_control_formal_1}\\
        b_1(\pi, \xi) &= \ul{s} + \left( \frac{\xi}{\beta} \right)^{-1} \left( \left( \frac{\pi}{1- \beta} \right)^{1-\beta} \left( \frac{\xi}{\beta} \right)^{\beta}  \right)^{-\wh{\gamma}}\ .
        \label{eq:opt_control_formal_2}
    \end{align}
    The second order condition guarantees that these controls are a local maximum if $0 < \beta < 1$ and $\gamma > 1$, which we have imposed.
    
    At this juncture, given the homogeneity of the utility function, it is natural to guess a function of the form \eqref{eq:g_expr}, substitute into \eqref{eq:hjb_g_ustar}, and solve for the constant of integration as
    \begin{equation}
        \Gamma = \left( \frac{(P / \wt{\eta})^\beta}{(1-\beta )^{1 - \beta} \beta^{\beta}} \right)^{\wh{\gamma}} \varphi > 0\ ,
    \label{eq:def_gamma}
    \end{equation}
    where
    \begin{equation}
        \varphi = \frac{\gamma  \left(\kappa ^2 + 2 \wh{\delta} -2 (\gamma -1) \mu_R \right)-\kappa ^2}{2 \gamma ^2} > 0\ .
    \label{eq:def_varphi}
    \end{equation}
    The controls \eqref{1d:eq:opt_controls_tg} follow. That these are admissible can be seeing by noting that $\diff \wt{Z}_t = \diff \wt{W}_t$ and inserting the optimal controls to obtain the geometric Brownian motion \eqref{1d:eq:opt_z_tilde_t}, which is always positive for a positive initial condition. With $\mu_{\wt{Z}}$ and $\sigma_{\wt{Z}}$ the drift of the process in \eqref{1d:eq:opt_z_tilde_t}, note that we have the identity
    \begin{equation}
        (1 - \gamma) (\mu_{\wt{Z}} - \tfrac{1}{2} \gamma \sigma_{\wt{Z}}^2) = \wh{\delta} - \varphi\ ,
    \end{equation}
    as can be checked. The transversality condition therefore computes as
    \begin{equation}
        \lim_{T \to \infty} \me^{- \wh{\delta} T} \E\left[ G(\wt{W}_T^{w}) \right] = \lim_{T \to \infty} \Gamma^{-\gamma} \frac{\wt{z}(w)^{1-\gamma} \me^{-\varphi T}}{1-\gamma} = 0\ .
    \end{equation}
    Hence, the conditions for the verification theorem for this candidate solution \citep[cf.][Thm.~3.5.3]{Pham2009} are satisfied.
\end{proof}

\begin{proof}[Proof of Theorem~\ref{1d:thm:val_func}]
    The proof is in three steps. Firstly, the control-stopping problem is transformed into a pure stopping problem. In a following step, a characterisation of the stopping region is obtained, which yields the claimed stopping rule in Case (i). Finally, the decision problem is solved for Case (ii).

    \emph{(a) From control-stopping to pure stopping.} Firstly, note that at the moment of investment, we have the identity
    \begin{equation}
        \wt{Z}_\tau - Z_\tau = \theta\ .
    \label{1d:eq:z_zt_theta}
    \end{equation}    
    It follows from Proposition~\ref{1d:prop:terminal_gain} that the reward received upon stopping is given by
    \begin{equation}
        \Gamma^{-\gamma}\frac{(Z_\tau^{z(w)} + \theta)^{1-\gamma}}{1 - \gamma}\ .
    \end{equation}
    Hence, following a standard argument, associate the value function of \eqref{1d:eq:val_func} to the following ordinary differential equation in variational form:
    \begin{equation}
        \min \Big[ \wh{\delta} F - \sup_{a, x, s} \left[ \mathcal{L}^{a,x,s}_{w} + U(x, s) \right], \quad F(w) - \Gamma^{-\gamma} u(z(w) + \theta) \Big] = 0\ ,
    \label{eq:hjb_vi_f}
    \end{equation}
    where
    \begin{equation}
       u(x) \defeq \frac{x^{1-\gamma}}{1- \gamma}
    \label{1d:eq:util_crra}
    \end{equation}
    and 
    \begin{equation}
        \mathcal{L}^{a, x, s}_w F =  ( a \kappa \sigma_S w + \mu_R w + Y - x - (s / \eta) P ) \partial_w F + \frac{1}{2} a^2 \sigma_S^2 w^2 \partial_w^2 F
    \label{eq:lagrange_f}
    \end{equation}
    is the infinitesimal generator of wealth before investment. Now, as in the proof of Proposition~\ref{1d:prop:terminal_gain}, assume $\partial_w F > 0$ and $\partial_w^2 F < 0$ and perform the formal optimisations over the controls in \eqref{eq:hjb_vi_f}. Then change variables by defining $f(z) \defeq F(w(z))$ and $g(z) \defeq \Gamma^{-\gamma} u(z + \theta)$ to transform \eqref{eq:hjb_vi_f} to
     \begin{equation}
        \min \Bigg[ \wh{\delta} f - \Bigg( \mu_R z \partial_z f - \frac{1}{2}\frac{\kappa^2 (\partial_{z} f)^2}{\partial_{z}^2 f} + \wh{U}(\partial_{z} f,  (P / \eta) \partial_{z} f) \Bigg)\ , \quad f - g \Bigg] = 0\ .
    \label{eq:hjb_vi_z}
    \end{equation}
    Define now the Legendre-Fenchel transform
    \begin{equation}
        \wh{f}(\wh{z}) = \sup_{z > 0}  \left[ f(z) - z \wh{z}  \right]\ ,
    \end{equation}
    whence the identities
    \begin{equation}
        \partial_z f = \wh{z}\ ,\quad \partial_z^2 f = - (\partial_{\wh{z}}^2 \wh{f})^{-1}
    \label{eq:legendre_f_identities}
    \end{equation}
    follow. Similarly define and compute
    \begin{equation}
        \wh{g}(\wh{z}) \defeq \sup_{z > 0} \left[ g(z) - z \wh{z}  \right] = \Gamma^{-1} \wh{u}( \wh{z}) + \theta \wh{z}\ ,
    \label{eq:def_gstarhat}
    \end{equation}
    where 
    \begin{equation}
        \wh{u}({\wh{x}}) = \sup_{x} \left[  u(x) - x \wh{x} \right] = \frac{\wh{x}^{-\wh{\gamma}}}{\wh{\gamma}}
    \end{equation}
    for $\wh{\gamma} = (1 - \gamma) / \gamma$. Insert these identities into \eqref{eq:hjb_vi_z} to obtain
    \begin{equation}
        \min \bigg[ \wh{\delta} \wh{f} - \left( (\wh{\delta} - \mu_R) \wh{z} \partial_{\wh{z}} \wh{f} + \frac{1}{2} \kappa^2 \wh{z}^2 \partial_{z}^2 \wh{f} + \wh{U}(\wh{z},  (P / \eta) \wh{z}) \right), \quad \wh{f} - \wh{g} \bigg] = 0\ .
    \label{1d:eq:hjb_dual}
    \end{equation}
    This is associated to the optimal stopping problem
     \begin{equation}
        \wh{f}(\wh{z}) = \sup_\tau \E \Bigg[ \int_0^\tau \me^{-\wh{\delta} t} \wh{U}(\wh{Z}_t^{\wh{z}},  (P / \eta) \wh{Z}_t^{\wh{z}}) \diff t + \me^{-\wh{\delta} \tau} \wh{g}(\wh{Z}_\tau) \Bigg]\ ,
    \label{eq:def_fhat_stop}
    \end{equation}
    where the dual process $\wh{Z}_t$ follows the geometric Brownian motion
    \begin{equation}
        \diff \wh{Z}_t = (\wh{\delta} - \mu_R) \wh{Z}_t \diff t + \kappa \wh{Z}_t \diff B_t\ .
    \end{equation}
    This decision problem is very similar to the one studied by \citet[Ch.~2.16]{Rogers2013}, whose method we follow in part (c) of the proof below.

    \emph{(b) Characterisation of the stopping region, solution when $\theta \geq 0$.} The continuation and stopping regions associated to \eqref{eq:def_fhat_stop} are given by
    \begin{align}
        \mathcal{C} &\defeq \{ \wh{z} \in \R_+ \mid \wh{f}(\wh{z}) > \wh{g}(\wh{z}) \}\ ,\\
        \mathcal{S} &\defeq \{ \wh{z} \in \R_+ \mid \wh{f}(\wh{z}) = \wh{g}(\wh{z}) \}\ .
    \end{align}
    Further, with $\mathcal{L}_{\wh{z}}$ the infinitesimal generator of the process $\wh{Z}_t$, it is a standard result that
    \begin{alignat}{2}
        \mathcal{C} &\supset \mathcal{D} &&\defeq \{ \wh{z} \in \R_+ \mid \wh{\delta} \wh{g} - (\mathcal{L}_{\wh{z}}\,\wh{g} + \wh{U}(\wh{z},  (P / \eta) \wh{z})) < 0 \}\ ,\\
        \mathcal{S} &\subset \mathcal{D}^c &&\defeq \{ \wh{z} \in \R_+ \mid \wh{\delta} \wh{g} - (\mathcal{L}_{\wh{z}}\,\wh{g} + \wh{U}(\wh{z},  (P / \eta) \wh{z})) \geq 0 \}\ ,
    \end{alignat}
    with $\tau^* = 0$ being optimal in the stopping problem if $\mathcal{D} = \emptyset$ \cite[Prop.~3.4]{Oeksendal2019}. The condition defining the set $\mathcal{D}^c$ computes as
    \begin{equation}
        \gamma  \varphi ((\wh{\eta} / \eta)^{-\beta \wh{\gamma}} -1)  (z + \theta) + (\gamma -1) \theta \mu_R \geq 0\ ,
    \end{equation}
    which is satisfied if $\theta \geq 0$. Hence, in this case $\mathcal{D} = \emptyset$ and $\tau^* = 0$ is optimal. It follows that $\wh{f} = \wh{g}$, which is true only if and only if $F(w) = \Gamma^{-\gamma} u(z(w) + \theta)$. Thus, investing immediately is optimal, and Proposition~\ref{1d:prop:terminal_gain} applies as claimed.

    \emph{(c) Solution when $\theta < 0$.} Suppose now that $\theta < 0$ so that we have to solve \eqref{eq:def_fhat_stop} in generality. It is in fact a standard optimal stopping problem in one dimension, albeit in dual space. Hence, the investment threshold is of the form $\wh{z} \leq \wh{z}^*$, that is, the agent invests when the marginal utility of wealth is small enough. We thus conjecture a general solution to \eqref{1d:eq:hjb_dual} of the form
    \begin{equation}
        \wh{f}(\wh{z}) = \begin{cases}
            \wh{g}(\wh{z})\ ,&\wh{z} \leq \wh{z}^*\ ,\\
            \Phi^{-1} \wh{u}(\wh{z}) + A_0 (\wh{z} / \wh{z}^*)^{-a_0} + A_1 (\wh{z} / \wh{z}^*)^{a_1}\ ,& \wh{z} \geq \wh{z}^*\ ,
        \end{cases}
    \label{1d:eq:f_hat}
    \end{equation}
    where $\Phi > 0$ is identical to $\Gamma$ of \eqref{eq:def_gamma} with $\wt{\eta}$ replaced by $\eta$, $A_0$ and $A_1$ are constants of integration, and $-a_0 < 0 < 1 < a_1$ are the roots of the following quadratic equation in $x$:
    \begin{equation}
        \frac{1}{2}\kappa^2 x (x -1) + (\wh{\delta} - \mu_R) x - \wh{\delta} = 0\ .
    \end{equation}
    The convexity of $\wh{f}$ demands that we set $A_1$ to zero; it remains to solve for $A_0$ and $\wh{z}^*$ from the smooth pasting conditions:
    \begin{equation}
        \wh{f}(\wh{z}^*) = \wh{g}(\wh{z}^*)\ ,\quad \wh{f}'(\wh{z}^*) = \wh{g}'(\wh{z}^*)\ .
    \end{equation}
    It follows that
    \begin{align}
        A_0 &= \left(-\frac{(a_0  \gamma +\gamma -1) (\Gamma -\Phi )^{\frac{\gamma }{\gamma -1}} ((a_0 +1) (\gamma -1) \Gamma  \Phi )^{\frac{\gamma }{1-\gamma }}}{\theta }\right)^{\gamma -1}\ ,\\
        \wh{z}^* &= \left(-\frac{(a_0  \gamma +\gamma -1) (\Gamma -\Phi )}{(a_0 +1) (\gamma -1) \Gamma  \theta  \Phi }\right)^{\gamma }\ , 
    \end{align}
    completing the solution in dual space. The investment trigger in primal space follows by again invoking smooth pasting:
    \begin{equation}
        \wh{z}^* = f'(z^*) = g'(z^*) = \Gamma^{-\gamma} (z^* + \theta)^{-\gamma}\ ;
    \end{equation}
    inserting the definitions of $\Gamma$ and $\Phi$ gives the identity
    \begin{equation}
        z^* = \frac{1}{a_0 \gamma + \gamma - 1} \left( \frac{(1 + a_0)(\gamma - 1)}{(\wt{\eta}/\eta)^{\beta \wh{\gamma}} - 1} - a_0 \right) \theta \invdefeq \Lambda \theta\ ,
    \label{1d:eq:trig_zstar}
    \end{equation}
    where $\Lambda < 0$ so that $z^* > 0$. The claim follows.
\end{proof}

\begin{proof}[Proof of Proposition~\ref{1d:prop:rebound}]
    \emph{Case (i).} If $\theta > 0$, investment is immediate. It follows from the optimal controls in Proposition~\ref{1d:prop:terminal_gain} and Corollary~\ref{1d:cor:counterfactual} that the rebound measure is given by the geometric Brownian motion
    \begin{equation}
        R_t = \frac{\beta \varphi}{P} \left( \wt{\eta} (z + \theta) - \eta z \right) \exp{\left[ (\mu_{\wt{Z}} - \tfrac{1}{2}\sigma_{\wt{Z}}^2) t + \sigma_{\wt{Z}} B_t \right]}\ ,\quad t \geq 0\ .
    \label{1d:eq:rebound_calc_imm}
    \end{equation}
    This quantity is strictly positive if $\theta > 0$. Hence $\E [R_t] > 0$ for all $t \geq 0$ as claimed. Similarly, the backfire measure is computed as
    \begin{equation}
        Q_t = (\wt{\eta}^{-1} - \eta^{-1}) \ul{s} + \frac{\beta \varphi}{P} \theta \exp{\left[ (\mu_{\wt{Z}} - \tfrac{1}{2}\sigma_{\wt{Z}}^2) t + \sigma_{\wt{Z}} B_t \right]}\ ,\quad t \geq 0\ .
    \label{1d:eq:backfire_calc_imm}
    \end{equation}
    It follows that $Q_t \geq 0$ if
    \begin{equation}
        \vartheta_t^\theta \geq \frac{(\eta^{-1} - \wt{\eta}^{-1}) \ul{s} P}{\beta \varphi} \defeq \kappa\ ,
    \label{1d:eq:def_kappa}
    \end{equation}
    where we denote
    \begin{equation}
        \vartheta_t^\theta \defeq \theta \exp{\left[ (\mu_{\wt{Z}} - \tfrac{1}{2}\sigma_{\wt{Z}}^2) t + \sigma_{\wt{Z}} B_t \right]}\ ,\quad t \geq 0\ .
    \label{1d:eq:vartheta_t}
    \end{equation}
    Since $\vartheta_t$ has a lognormal distribution at time $t$, the claim \eqref{1d:eq:prob_backfire} follows.

    \emph{Case (ii).} If $\theta = 0$ exactly, from the expressions \eqref{1d:eq:rebound_calc_imm} and \eqref{1d:eq:backfire_calc_imm} above it is immediate that $\E[R_t] > 0$ and $\E[Q_t] < 0$.

    \emph{Case (iii).} If $\theta < 0$ and $z \geq z^*$, investment is again immediate, and the expression \eqref{1d:eq:rebound_calc_imm}  for $R_t$ holds. However, since $\theta < 0$, the requirement that $R_t > 0$ reduces to the requirement that the initial condition of the geometric Brownian motion in \eqref{1d:eq:rebound_calc_imm} be positive. The claim \eqref{1d:eq:rebound_pos_req}
    follows. Similarly, since the expression \eqref{1d:eq:backfire_calc_imm} for $Q_t$ holds in this case, It is immediate that $\E[Q_t] < 0$ if $\theta < 0$.

    \emph{Case (iv).} If $\theta < 0$ and $z < z^*$, waiting is optimal. So with $\wh{\tau}^*$ the optimal investment time in Approximation~\ref{1d:assump_approx}, define the shifted Brownian motion
    \begin{equation}
        B_t^{\wh{\tau}^*} \defeq B_t - B_{\wh{\tau}^*}\ ,\quad t \geq {\wh{\tau}^*}\ . 
    \label{1d:eq:brownian_shift}
    \end{equation}
    Using now the controls in Approximation~\ref{1d:assump_approx} for the investment scenario and Corollary~\ref{1d:cor:counterfactual} for the counterfactual scenario, apply the definition of rebound in \eqref{1d:eq:def_rebound} to obtain
    \begin{equation}
        R_t =
        \begin{dcases}
            0\ , &t < \wh{\tau}^*\ ,\\
            \frac{\beta \varphi}{P} \left( \wt{\eta} (z^* + \theta) - \eta z^* \right) \exp{\left[ (\mu_{\wt{Z}} - \tfrac{1}{2}\sigma_{\wt{Z}}^2) (t - \wh{\tau}^*) + \sigma_{\wt{Z}} B_t^{\wh{\tau}^*} \right]}\ ,&t \geq \wh{\tau}^*\ .
        \end{dcases}
    \end{equation}
    Similarly to Case (iii) and \eqref{1d:eq:rebound_pos_req}, demanding the non-negativity of $R_t$ and inserting the identity $z^* = \Lambda \theta$ from \eqref{1d:eq:trig_zstar} yields the claim. On the other hand, the backfire measure computes as
    \begin{equation}
        Q_t =
        \begin{dcases}
            0\ ,&t < \wh{\tau}^*\ ,\\
            (\wt{\eta}^{-1} - \eta^{-1}) \ul{s} + \frac{\beta \varphi}{P} \vartheta_t^{\theta; \wh{\tau}^*}\ ,&t \geq \wh{\tau}^*\ ,
        \end{dcases}
    \label{1d:backfire_calc_wait}
    \end{equation}
    where $\vartheta_t^{\theta; \wh{\tau}^*}$ is identical to \eqref{1d:eq:vartheta_t} except with driver $B_t^{\wh{\tau}^*}$ instead of $B_t$. As in Case (iii), the claim follows since $\theta < 0$.
\end{proof}

\begin{proof}[Proof of Theorem~\ref{1d:thm:welfare_improve}]
    We first define and compute the following integral using Fubini's theorem, utilising the fact that we have the product of independent geometric Brownian motions:
     \begin{align}
        \mathcal{I}(\pi)
        &\defeq \E \left[ \int_0^\infty \me^{-\wh{\epsilon} t} \varpi_t^\pi Q_t \diff t \right]\\
        &= \left( \frac{(\wt{\eta}^{-1} - \eta^{-1}) \ul{s}}{\wh{\epsilon} - \mu_\varpi} + \frac{\beta \varphi\theta}{(\wh{\epsilon} - \mu_\varpi - \mu_{\wt{Z}}) P} \right) \pi\ ,
    \label{1d:eq:I2}
    \end{align}
    where $\wh{\epsilon} - \mu_\varpi - \mu_{\wt{Z}} > 0$ is required for convergence. Suppose now we are in Case (i). Then $\tau^*(z) = 0$, and it follows from \eqref{1d:eq:social_cost} that $V_{\mathrm{sc}}(z; \pi) = \mathcal{I}(\pi)$. Equation~\ref{1d:eq:theta_bnd} follows by algebraic manipulations.

    For Case (ii), note firstly that due to Approximation~\ref{1d:assump_approx}, the integral in \eqref{1d:eq:social_cost} in fact has lower limit $\wh{\tau}^*$, since $Q_t = 0$ for $t < \wh{\tau}^*$ as shown in \eqref{1d:backfire_calc_wait}. Consequently, change the integration variable $t \mapsto t + \wh{\tau}^*$ and apply the strong Markov property; it follows that 
    \begin{equation}
        V_{\mathrm{sc}}(z; \pi) =
        \E \left[
        \me^{-\wh{\epsilon}\wh{\tau}^*(z)}
        \mathcal{I}(\varpi_{\wh{\tau}^*(z)})
        \ \Big|\ (\wh{Z}_0; \varpi_0) = (z; \pi) \right]\ .
    \end{equation}
    Since $\mathcal{I}$ is linear in its argument and $\varpi_t$ is a geometric Brownian motion, we have
    \begin{equation}
        \E \left[ \me^{-\wh{\epsilon}\wh{\tau}^*(z)} \varpi_{\wh{\tau}^*(z)} \right]
        = \pi\, \E \left[ \me^{-(\wh{\epsilon} - \mu_\varpi)\wh{\tau}^*(z)} \right]
        \invdefeq \pi\, \mathscr{L}(z; \wh{\epsilon} - \mu_\varpi)\ ,
    \end{equation}
    where $\mathscr{L}$ is the Laplace transform of the distribution of the investment time. This yields the claimed identity \eqref{1d:eq:social_cost_wait}. Since $\wh{Z}_t$ is a geometric Brownian motion, the Laplace transform is explicit \citep[Ch.~3.3.2]{Jeanblanc2009}: for $z < z^*$ we have
    \begin{equation}
        \mathscr{L}(z; \varrho)
        = (z^*/z)^{a_0(\varrho)}\ ,
    \label{1d:eq:laplace_transform}
    \end{equation}
    where
    \begin{equation}
        a_0(\varrho) = \left(\nu_{\wt{Z}}- \sqrt{\nu_{\wt{Z}}^2 + 2 \varrho} \right) / \sigma_{\wt{Z}}\ ,\quad
        \text{for} \quad \nu_{\wt{Z}}\defeq (\mu_{\wt{Z}}- \sigma_{\wt{Z}}^2 / 2)/\sigma_{\wt{Z}}\ .
    \label{1d:eq:def_d0}
    \end{equation}
    The claim follows.
\end{proof}


\begin{proof}[Proof of Proposition~\ref{1d:prop:approx_subsidy}]
    Firstly, make a calculation analogous to the proof of Theorem~\ref{1d:thm:welfare_improve} to reduce the objective to the deterministic function
    \begin{equation}
        \begin{cases}
            \inf_m \left[ \mathcal{I}(\pi; m) + \Psi(m K) \right]\ ,&(\theta \geq 0)\lor(z \geq z^*)\ ,\\
            \inf_m \left[ \mathscr{L}(z; \wh{\epsilon} - \mu_\varpi) \mathcal{I}(\pi; m) + \mathscr{L}(z; \wh{\epsilon}) \Psi(m K) \right]\ ,&(\theta < 0)\land(z < z^*)\ .
        \end{cases}
    \label{1d:eq:j_z_init}
    \end{equation}
    To simplify notation, we group some constants. From \eqref{1d:eq:I2} write $\mathcal{I}(\pi) = A_0 + A_1 \theta$, where
    \begin{equation}
        A_0 \defeq \frac{(\wt{\eta}^{-1} - \eta^{-1})\ul{s} \pi}{\wh{\epsilon} - \mu_\varpi}\ , \quad
        A_1 \defeq \frac{\beta \varphi \pi}{P(\wh{\epsilon} - \mu_\varpi - \mu_{\wt{Z}})}\ ,
    \end{equation}
    and from \eqref{1d:eq:z_zt_theta} write $\theta = B_0 - B_1 K$, where
    \begin{equation}
        B_0 \defeq (\eta^{-1} - \wt{\eta}^{-1}) \ul{s} P / \mu_R\ , \quad
        B_1 \defeq \rho/\mu_R\ .
    \end{equation}
    Introduce now the subsidy by mapping $K \mapsto (1 - m)K$ and propagate through the above identities. It follows that we can write $\mathcal{I}(\pi; m) = C_0 + C_1 m$ for constants
    \begin{equation}
        C_0 \defeq A_0 + A_1 \theta\ ,\quad C_1 = A_1 B_1 K\ .
    \label{1d:eq:c_cons}
    \end{equation}
    Hence, if $\theta \geq 0$ or $z \geq z^*$, \eqref{1d:eq:j_z_init} reduces to
    \begin{equation}
        \inf_m \left[ (C_0 + C_1 m) + \Psi(m K) \right]\ .
    \end{equation}
    Case (i) now follows by direct computation, noting that the second-order condition $\partial_m J = \xi_1 K^2 > 0$ guarantees an infimum.
    
    For Case (ii), note that the subsistence net-present value in the presence of the subsidy, i.e.~$B_0 - B_1 (1 - m) K$, has a root at 
    \begin{equation}
        \ol{m} \defeq 1 - \frac{B_0}{B_1 K} < 1\ .
    \label{1d:eq:ol_m_def}
    \end{equation}
    This corresponds to the maximum allowable subsidy, since the agent invests immediately when the subsistence net-present value is non-negative. Then from \eqref{1d:eq:laplace_transform} and \eqref{1d:eq:trig_zstar} we have the identity
    \begin{equation}
        \mathscr{L}(z; \varrho) = (D_0(z) + D_1(z) m)^{a_0(\varrho)}\ ,
    \end{equation}
    where
    \begin{equation}
        D_0(z) \defeq \frac{\Lambda \theta}{z}\ ,\quad
        D_1(z) \defeq \frac{\Lambda B_1 K}{z}\ ,
    \end{equation}
    and $a_0(\varrho)$ is from \eqref{1d:eq:def_d0}. Hence, defining $d_0 \defeq a_0(\wh{\epsilon} - \mu_\varpi)$ and $d_1 \defeq a_0(\wh{\epsilon})$, the optimisation \eqref{1d:eq:j_z_init} reduces in the second case to \eqref{1d:eq:sp_final}. It follows that the optimal subsidy is given either by the interior solution or a boundary value; optimality must be verified by explicitly checking the value function $J$.
\end{proof}

\section{Analysis of approximate controls}
\label{append:base_model:approx_controls}

The aim of this section is a formal treatment of Approximation~\ref{1d:assump_approx}. Let $w \in (-H, w^*)$ be a relevant initial condition for the wealth process $W_t$ of \eqref{1d:eq:wt_evolve}, and let $\tau \in \mathcal{T}$ be a stopping time. With $\wh{F}$ the counterfactual value function from Corollary~\ref{1d:cor:counterfactual} and $G$ the terminal gain from Proposition~\ref{1d:prop:terminal_gain}, define
\begin{equation}
    g(w) \defeq \wh{F}(w) - G(w)\ .
\label{append:eq:def_small_g}
\end{equation}
Then with $\wh{a}$, $\wh{x}$, and $\wh{s}$ the optimal controls from Corollary~\ref{1d:cor:counterfactual}, apply Dynkin's formula to $\wh{F}$ to obtain
\begin{align}
    \E[ \me^{-\wh{\delta} \tau} \wh{F}(W_\tau^{w; \tau})] &= \wh{F}(w) + \E \left[ \int_0^\tau \me^{-\wh{\delta} t} (-\wh{\delta} + \mathcal{L}_w^{\wh{a}, \wh{x}, \wh{s}} ) \wh{F}(W_t^{w; \tau}) \diff t \right]\\
    &= \wh{F}(w) + \E_w \left[ \int_0^\tau \me^{-\wh{\delta} t} (-U(\wh{x}_t, \wh{s}_t)) \diff t \right]\ ,
\label{append:eq:dynkin_fhat}
\end{align}
where we use the fact that $\wh{F}$ solves an HJB equation. Now, let $(a, x, s, \tau) \in \mathcal{A}(w)$ be arbitrary admissible controls in the agent's decision problem. Adding the present value of utility up to the time $\tau$ with respect to these controls to both sides of \eqref{append:eq:dynkin_fhat} and inserting the definition \eqref{append:eq:def_small_g} yields
\begin{multline}
    \E_w \left[ \int_0^\tau \me^{-\wh{\delta} t} U(x_t, s_t) \diff t + \me^{-\wh{\delta} \tau} (G(W_\tau^{w; \tau}) + g(W_\tau^{w; \tau})) \right] = \\
    \wh{F}(w) + \E_w \left[ \int_0^\tau \me^{-\wh{\delta} t} \left( U(x_t, s_t) - U(\wh{x}_t, \wh{s}_t) \right) \diff t \right]\ .
\end{multline}
Taking now the supremum over $a$, $x$, $s$ and $\tau$, and applying the definition of the value function \eqref{1d:eq:val_func} thus gives the identity
\begin{equation}
    F(w) = \wh{F}(w) + \Omega(w)\ ,
\label{append:eq:val_func_rewrite}
\end{equation}
where
\begin{equation}
    \Omega(w) \defeq \sup_{a, x, s, \tau} \E_w \bigg[ \int_0^\tau \me^{-\wh{\delta} t} \left( U(x_t, s_t) - U(\wh{x}_t, \wh{s}_t) \right) \diff t\\
    - \me^{-\wh{\delta} \tau} g(W_\tau^{w; \tau}) \bigg]\ .
\end{equation}
Intuitively, \eqref{append:eq:val_func_rewrite} states that the agent's value function equals the value of continuing forever with the present technology, $\wh{F}$, plus the value of the option to switch technologies, given by $\Omega$. Note that since $F$ is continuous with $F(w) = G(w)$ for $w \geq w^*$ (see Theorem~\ref{1d:thm:val_func}), it must be the case that $\Omega$ is continuous as well, satisfying the boundary condition
\begin{equation}
    \Omega(w) = G(w) - \wh{F}(w)\ ,\quad w \geq w^*\ .
\label{append:eq:opt_val_asymp}
\end{equation}

The consequences of this for the optimal controls are immediate. Consider that the implicit expressions for the pre-investment controls in Theorem~\ref{1d:thm:val_func} depend on the first and second derivatives of $F$. Therefore, since the derivative is a linear operator, as long as the derivatives of $\Omega$ are small relative to the derivatives of $\wh{F}$, the counterfactual controls in \eqref{1d:eq:approx_controls} will provide a good approximation to the true controls. Intuitively, we expect the option value of switching $\Omega(w)$ to have relatively small derivatives for $w$ far away from the investment threshold $w^*$; conversely, as $w$ approaches the threshold, we expect the derivatives of $\Omega(w)$ to increase in relative size as the value of switching approaches the regime defined in \eqref{append:eq:opt_val_asymp}. The upshot is that Approximation~\ref{1d:assump_approx} is expected to be a fair approximation away from the threshold $w^*$, but to deteriorate in quality closer to the threshold.

Figure~\ref{fig:1d_error_controls} shows that this bears out for the Case Study in Section~\ref{1d:sec:case_study}. It displays the relative errors between Approximation~\ref{1d:assump_approx} and the numerically-estimated true controls. The graphic makes clear that although the relative error increases as $w \to w^*$, the size of the error is small; particularly for the energy-consumption control $s$, the largest error is only \num{15} basis points, and therefore does not meaningfully impact estimates for the welfare quantities $Q_t$, $V_\text{sc}$, and others. Moreover, although the approximation underestimates allocation, it overestimates both consumption controls, meaning that the results in Section~\ref{1d:sec:sol:opt_strategies} are indeed conservative, as claimed. Overall, the approximation appears robust, supporting the qualitative conclusions of Section~\ref{1d:sec:welfare:approx} and the numerical results in Section~\ref{1d:sec:case_study}.

\begin{figure}
    \centering
    \includegraphics[width=\linewidth]{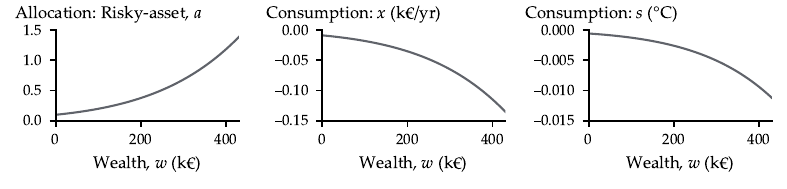}
    \caption{Relative error in \unit{\percent} in the optimal strategies if Approximation~\ref{1d:assump_approx} is assumed, with parameters from the case study in Section~\ref{1d:sec:case_study}. The errors are plotted over the range $w \in [0, w^*]$, where $w^* = \qty{430}{\kilo\euro}$ is the investment threshold.}
    \label{fig:1d_error_controls}
\end{figure}

\section{The agent's private gain}
\label{append:private_gain}

This appendix defines the agent's private gain within the context of the model. To this end, recall that the value function $F$ and counterfactual value $\wh{F}$ are the dynamic equivalents of indirect utility functions from microeconomics: they accept prices and wealth as arguments, and return the present value of lifetime utility \citep[cf.][Ch.~7.2]{Varian1992}. Consequently, the \emph{expenditure function} $\wh{F}^{-1}(u)$ encodes the amount of wealth required to attain a given level of lifetime utility in the counterfactual. The \emph{equivalent variation} of the energy-efficiency investment is then given by
\begin{equation}
    V_{\text{ev}}(w) \defeq \wh{F}^{-1}(F(w)) - w\ .
\end{equation}
This quantity gives the change in wealth which would be equivalent to the proposed change in utility due to the retrofit \citep[cf.][Ch.~10.1]{Varian1992}. Since it is measured in monetary terms, it can be directly compared to the social cost of the retrofit investment, $V_\text{sc}$, developed in Section~\ref{1d:sec:welfare:definitions}. In particular, it follows that the total change in welfare due to the retrofit is given by netting out the social cost from the agent's equivalent variation, which yields the measure
\begin{equation}
    V(w; \pi) \defeq V_{\text{ev}}(w) - V_{\text{sc}}(w; \pi)\ .
\label{1d:eq:welfare_tot}
\end{equation}
We have the following intuitive result, which slightly generalises some of the observations made in Section~\ref{1d:sec:welfare:definitions}.

\begin{prop}
\label{1d:prop:sw}
    In the absence of backfire, the energy-efficiency investment improves total welfare.
\end{prop}
\begin{proof}
    Notice that since $\tau^*$ is chosen optimally in \eqref{1d:eq:val_func}, and since the choice set includes the counterfactual scenario $\tau = \infty$, the equivalent variation $V_{\text{ev}}$ is necessarily non-negative. The claim follows since  $V_{\text{sc}} \leq 0$ if $Q_t \leq 0$ for all $t$.
\end{proof}

As a final application, we provide a closed-form expression for the agent's equivalent variation in the case of immediate investment; an analogous result is not possible when waiting is optimal, since a closed-form solution for $F(z)$ is not available in this case.

\begin{lem}
\label{1d:lem:ev}
    Suppose $\theta \geq 0$ or $z > z^*$. Then $V_{\mathrm{ev}}(z) = (\wt{\eta}/\eta)^\beta (z + \theta) - z$.
\end{lem}

\noindent We see that the monetary equivalent of the utility improvement due to the retrofit is an intuitive expression involving only the present level of disposable capital $z$, the ratio of the efficiency parameters $\wt{\eta}$ and $\eta$, the preference weight on the energy service $\beta$, and the improvement in human capital $\theta$.

\section{Optimal retrofit depth}
\label{append:base_model:depth}

This section extends the case study in Section~\ref{1d:sec:case_study} to consider the interesting question of optimal retrofit depth. Namely, supposing that the post-retrofit efficiency $\wt{\eta}$ can be modelled as a function of cost $K$, we ask what level of efficiency improvement is optimal for the agent. To this end, Table~\ref{1d:tab:retrofit_params} lists additional parameters from \citet[Case Study \enquote{EFH78}]{Galvin2024} corresponding to retrofits of progressively higher standards, here \enquote{Level 1}, \enquote{Level 2}, and \enquote{Level 3}, with Level 1 being the focus of Section~\ref{1d:sec:case_study}. The question of optimal retrofit depth is approached as follows: firstly, for each $K$ and corresponding $\wt{\eta}(K)$ in Table~\ref{1d:tab:retrofit_params} the decision problem \eqref{1d:eq:decision_prob} is solved to obtain a value function $F(K)$; then, the $K$ that maximises $F(K)$ is selected. For completeness, this procedure is repeated for the welfare change $V$. This is visually summarised in Figure~\ref{fig:1d_retrofit_depth} for an agent with parameters as in Table~\ref{1d:tab:const_params}. On the left, a standard logistic function approximating $\wt{\eta}(K)$ is fit to the available data points, and on the right, the corresponding value function $F(K)$ and welfare change $V(K)$ is displayed. In this instance, the optimal retrofit depth for the agent lies between Levels 1 and 2; moreover, since backfire is absent here, social and private optimality coincide.\footnote{We note in passing that a treatment of the optimal subsidy problem with flexible retrofit depth is beyond the scope of the case study. This is due to the fact that the bilevel optimisation \eqref{1d:eq:opt_subsidy_def} rests on the planner's knowledge of the agent's decision rules, which would have to be extended to account for flexibility in investment size.}

\begin{table}[htbp]
    \footnotesize
    \centering
    \caption{Parameters for the study of optimal retrofit depth in Appendix~\ref{append:base_model:depth}, estimated from \citep{Galvin2024}. The breakdown of costs $K = K_{\mathrm{min}} + K_{\mathrm{ee}}$ is included since it is used for fitting the curve $\wt{\eta}(K_{\mathrm{ee}})$ in Figure~\ref{fig:1d_retrofit_depth}. The so-called \enquote{anyway} costs $K_{\mathrm{min}}$ are the minimum necessary costs for renovation, with $K_{\mathrm{ee}}$ being the additional costs for the energy-efficiency measures.}
    \begin{tabular}{c l r r r r}
        \toprule
        Parameter & Description (Unit) & Existing state & Level 1 & Level 2 & Level 3\\
        \midrule
        $\eta$ (resp.~$\wt{\eta}$) & Efficiency (\unit{\celsius\per\watt}) & 0.005 & 0.025 & 0.030 & 0.039\\
        $K$     & Retrofit cost (\unit{\kilo\euro}) & & 120 & 125 & 137\\
        \midrule
        $K_{\text{min}}$ & Anyway costs (\unit{\kilo\euro}) &  & 57 & 57 & 57\\
        $K_{\text{ee}}$ & Cost of energy-efficiency measure (\unit{\kilo\euro}) &  & 63 & 68 & 80\\
        \bottomrule
    \end{tabular}
    \label{1d:tab:retrofit_params}
\end{table}

\begin{figure}[htbp]
    \centering
    \includegraphics[width=\linewidth]{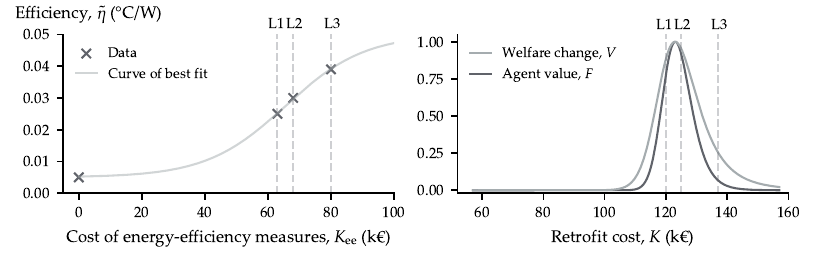}
    \caption{On the left: the data points from Table~\ref{1d:tab:retrofit_params}, with Level 1 abbreviated as \enquote{L1}, etc.~together with the curve of best fit, assuming a standard logistic function. On the right: the corresponding value function $F$ as well as the welfare change $\wh{V}$ are displayed as functions of $K$, each normalised to a $[0, 1]$ scale.}
    \label{fig:1d_retrofit_depth}
\end{figure}

\end{appendices}

\clearpage
\bibliographystyle{apalike-ejor}
\bibliography{biblio}

\end{document}